\documentclass[journal]{IEEEtran}
\usepackage{amsmath,amssymb}
\usepackage[dvips]{graphicx}
\usepackage{amsfonts}
\usepackage[mathscr]{eucal}
\usepackage{latexsym}
\usepackage{amsthm}
\usepackage{exscale}
\usepackage[mathscr]{eucal}
\usepackage{bm}
\usepackage[dvipsnames]{color}
\usepackage{cases}
\usepackage{epsfig}
\usepackage[center,small]{caption}
\usepackage{algorithm}
\usepackage{algorithmic}
\usepackage[verbose,nospace,sort]{cite}
\usepackage{tabularx}
\usepackage{multirow}
\usepackage{balance}
\usepackage{url}
\usepackage{graphicx}
\usepackage{subcaption}
\usepackage[export]{adjustbox}

\scrollmode

\newtheorem{lemma}{Lemma}

\IEEEoverridecommandlockouts

\title{Channel Knowledge Map for Cellular-Connected UAV via Binary Bayesian Filtering} 

\author{ Yuhang~Yang,~
	 Xiaoli~Xu,~\IEEEmembership{Member,~IEEE,}
	 Yong~Zeng,~\IEEEmembership{Senior Member,~IEEE,}
	 Haijian~Sun,~\IEEEmembership{Member,~IEEE,}
  and~Rose~Qingyang~Hu,~\IEEEmembership{Fellow,~IEEE}
	\thanks{
		Y. Yang, X. Xu and Y. Zeng are with National Mobile Communications Research Laboratory, Southeast University, Nanjing 210096, China. Y. Zeng is also with the Purple Mountain Laboratories,
		Nanjing 211111, China (e-mails: \{yuhang\_yang, xiaolixu,  yong\_zeng\}@seu.edu.cn). (Corresponding author: Yong Zeng.)
	}
	\thanks{H. Sun (e-mail: hsun@uga.edu) is with School of Electrical and Computer Engineering, University of Georgia, Athens, GA, USA.}
	\thanks{Rose~Qingyang~Hu (e-mail: rose.hu@usu.edu) is with Department of Electrical and Computer Engineering, Utah State University, Logan, UT, USA.}
	\thanks{This work was supported by the National Natural Science Foundation of China with grant number 62071114 and 62101118.}
	\thanks{Part of this work has been presented at the 2023 IEEE ICC Workshops,
		Rome, Italy in June 2023 \cite{yangsrapm}.}
}

\begin{document}

\maketitle

\begin{abstract}
Channel knowledge map (CKM) is a promising technology to enable environment-aware wireless communications and sensing.
Link state map (LSM) is one particular type of CKM that aims to learn the location-specific line-of-sight (LoS) link probability between the transmitter and the receiver at all possible locations, which provides the prior information to enhance the communication quality of dynamic networks. 
This paper investigates the LSM construction for cellular-connected unmanned aerial vehicles (UAVs) by utilizing both the expert empirical mathematical model and the measurement data. Specifically, 
 we first model the LSM as a binary spatial random field and its initial distribution is obtained by the empirical model. 
Then we propose an effective binary Bayesian filter to sequentially update the LSM by using the channel measurement. To efficiently update the LSM, we establish the spatial correlation models of LoS probability on the location pairs in both the distance and angular domains, which are adopted in the Bayesian filter for updating the probabilities at locations without measurements. 
Simulation results demonstrate the effectiveness of the proposed algorithm for LSM construction, which significantly outperforms the benchmark scheme, especially when the  measurements are sparse.


\end{abstract}

\pagestyle{plain}

\section{Introduction}

With future wireless networks trying to exploit high-frequency bands such as millimeter wave (mmWave) and TeraHertz (THz) bands, the blockage effect of wireless channel becomes more serious, especially for the air-to-ground (A2G) channel which largely relies on the line-of-sight (LoS) link \cite{zeng2019accessing}.  
On the other hand, compared with the ground links, the A2G channel is more predictable due to the strong coupling between the channel and the physical environment, together with the higher predictability of the aerial nodes' locations such as unmanned aerial vehicles (UAVs) \cite{Zeng2021,LAP}. To ensure safe flight and support a variety of critical missions, UAVs are typically equipped with communication and sensing devices. Besides, LoS A2G links are usually exploited to achieve high-rate and low-latency communications, such as video streaming communication \cite{10173754}, relay deployment optimization \cite{zheng2024geography}, and predictive beamforming \cite{10278781,zeng2023ckm}. 
Specifically, \cite{10278781} and \cite{zeng2023ckm} pointed out that the prior information about LoS link can enhance the sensing-assisted non-LoS (NLoS) identification by the echo signal and enable the predictive beamforming of mmWave communication links.
In addition, LoS link information is also shown to be useful for network localization \cite{mazuelas2018soft, wang2023indoor} and navigation \cite{yin2022millimeter} etc,  if such information can be shared within the network. 
On the other hand, the rapid growth of aerial users and the powerful data mining capabilities of wireless networks make it possible to realize 
\textit{environment-aware communication and sensing}
 for UAV using location-specific historical data. 
To this end, one promising technique is to leverage \textit{channel knowledge map (CKM)} \cite{Zeng2021,zeng2024tutorial}, which is a site-specific channel knowledge database with abundant location-specific wireless channel data and/or physical environment sensing data.
CKM is able to provide the intrinsic channel knowledge that is independent of the transmitter and receiver activities, based on the virtual or physical locations of the mobile terminals. For example, time of arrival (TOA) and angle of arrival (AOA) \cite{wuckm} between transmitter and receiver can be regarded as the typical channel knowledge.

However, efficient CKM construction is a non-trivial task. Most research on the construction of CKM focuses on channel gain map (CGM)  \cite{  chen2017efficient,zeng2021simultaneous,huang2021simultaneous,denkovski2012reliability,dall2011channel,li2022channel,chen2017learning,xu2024much}, while other categories of CKM remain less explored. 
For example, link state is one of the simplest channel knowledge, as it indicates the LoS/NLoS condition of the communications link between the transmitter and receiver. 
In \cite{yangsrapm}, we have proposed to construct link state map (LSM), CGM and physical map simultaneously by fusing sensing and channel data. However, acquiring precise sensing data requires dedicated sensing equipment, which is usually costly. It is desired to explore the efficient LSM construction algorithms based on limited and readily-available communication measurements.
As shown in Fig.~\ref{F:Introsummary}, there are two straightforward approaches for LoS link prediction, namely \textit{physical map-based } \cite{stump2011visibility,casarini2016updating,yin2022wireless,esrafilian2021map,linsalata2022map} and \textit{stochastic channel models-based} \cite{ortega2021nlos,schwarz2016gaussian,huang2020machine}. 
For the former, if we know the 3D physical map, we can compute the LSM by determining whether the line segment between the transmitter and receiver is blocked by obstacles or not. For the latter, the LSM can be constructed based on the stochastic channel models. Such models only exploit the large-scale environment information, such as the building density, environment type, user equipment (UE) height etc.
As a result, this method usually gives very coarse prediction, which can not meet the stringent accuracy requirements in some specific scenarios.

\begin{figure}[htb]
	\centering
	\includegraphics[scale=0.5]{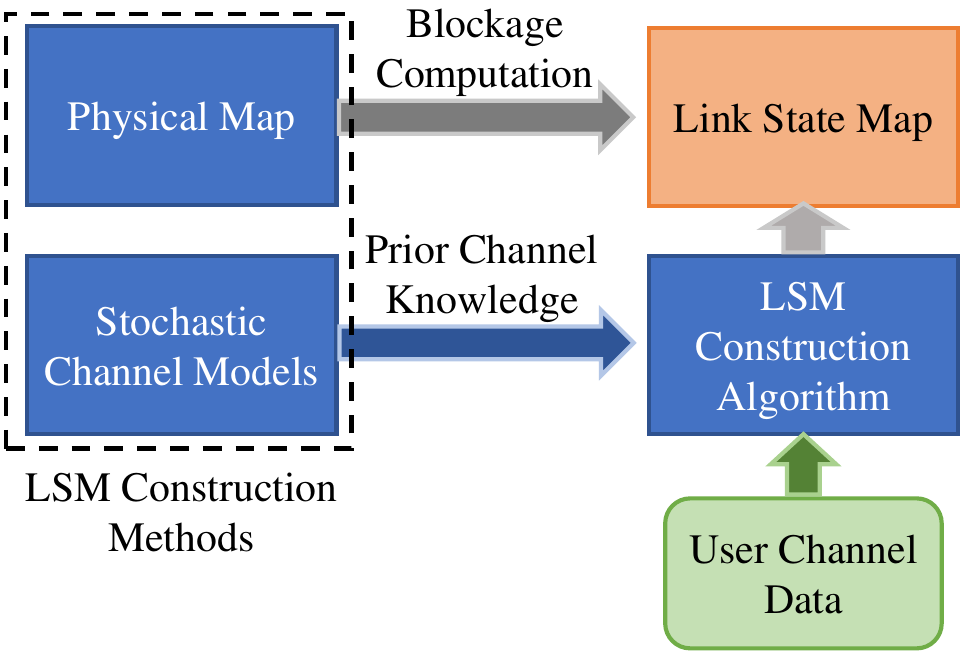}
	\caption{An overview of the LSM construction methods for cellular-connected UAVs.}
	\label{F:Introsummary}
\end{figure}

In the first approach, the physical map would be constructed when the user or base station (BS) has sufficient sensors, such as cameras and/or LiDAR.
With the known physical map, we can compute the blockage of LoS link. However, the computation cost from the physical map to LSM may be unaffordable if the map size is large \cite{bahoo2020computing,bahoo2019time}. 
The authors in \cite{stump2011visibility,casarini2016updating} tried to construct the LSM between any two locations based on the constructed physical map with the ranging sensors by performing blockage computation.  
In addition, they pointed out that the multi-robot can improve their physical map construction efficiently by optimizing the strategy of cooperating sensing if the robot network builds the LSM when performing missions online. 
Furthermore, the authors in \cite{yin2022wireless} addressed the problem of how to estimate the propagation when only a portion of the map has been reconstructed. 
The link states for those transmitter-receiver pairs outside the observed environment are estimated based on pre-trained machine learning algorithm.



For the second approach, LoS/NLoS identification based on stochastic models has been explored for decades. It is worth mentioning that extensive works \cite{3GPP,LAP,holis2008elevation} 
had given a universal probabilistic model, which provided a LoS probability with respect to high-level environmental parameters, such as the transmitter and receiver location, communication environment etc. For example, 3GPP TR 36.777  \cite{3GPP} provided a distance-dependent LoS probability model
between cellular-connected UAV and ground base station (GBS).
In addition, based on the Urban Micro (UMi) scenario of the 3GPP 3D channel model, reference \cite{schwarz2016gaussian} converts these correlated stochastic channel model into deterministic binary LoS/NLoS maps by applying appropriate thresholding. Reference \cite{ortega2021nlos} formulated a LoS/NLoS binary Bayesian classification problem for LSM construction and analyzed the performance of different estimators under the random geometry theory given a set of random blocking elements with a specific spatial density. In addition, the authors in \cite{huang2020machine} investigated the effectiveness of machine learning-enabled LoS/NLoS identification in dynamic environments. 



However, both the physical map-based and stochastic channel models-based LSM constructions face certain limitations.  On one hand, accurate physical map is difficult to obtain, and using the high accuracy environment sensor for constructing physical map usually incurs large cost. In addition, the computational overhead from occupancy grid map to LSM may be enormous. On the other hand, the main challenge of LSM construction based on stochastic channel model is that the wireless environment may vary significantly even within the same site, and hence the stochastic model may not match with the real environment. 
Besides, the real-time update of LSM is crucial to reflect the environment dynamics. The stochastic models are useful for LSM initialization and the channel measurements are important to capture the local environment characteristics.

LSM construction and update can be modeled as the binary estimation problem, which can be addressed by
binary Bayesian filter.
Binary estimation problems of this type arise if we intend to estimate a fixed binary quantity in the environment from a sequence of sensor measurements. These binary estimation problems have been investigated mostly in robotics and sensing \cite{thrun2002probabilistic}.
For example, during radar sensing, we might want to know if a target exists or not. Another example of binary Bayesian filters is that the robot determines if a door is open or closed by the ranging sensor.
 To make full use of the well-established LoS models and the measurement data, we propose the probabilistic LSM construction with radio propagation semantics. 
The main contributions of the paper can be summarized as follows:
\begin{itemize}
	\item 
	First, we consider the scenario where a cellular-connected UAV communicates with the GBS in an urban environment.
	We model the LSM as a binary spatial random field, and each binary random variable at a certain location represents the existence of LoS link with GBS.
	Then, we propose a novel sequentially probabilistic LSM construction approach based on a limited number of radio measurements and binary Bayesian filter, to update the prior LSM that has an empirical probability distribution given by expert knowledge. 
%
	\item Second, to further resolve the sequential update of LSM based on radio measurements and binary Bayesian filter, 
	we show that the LSM update exploiting each radio measurement can be divided into two cases, \textit{the measurement location} and \textit{unmeasured location}.
	Then, we derive the spatial correlation expression for the LoS probability. 
	Finally, based on this spatial correlation and a binary Bayesian filter, the LSM update is enhanced in each direction with measurements and in the relevant direction without measurements.
	We conduct extensive numerical simulations to validate the superiority of the proposed approach in probabilistic link state mapping. The simulation results demonstrate that the proposed map update algorithm can significantly improve the map construction accuracy when comparing with the benchmark such as $ K $-nearest neighbours (KNN) interpolation.
\end{itemize}

The rest of the paper is organized as follows. The LSM construction problem is formulated in Section II. Section III presents the spatial correlation of LoS link condition and provides the binary Bayesian filter-based sequential probabilistic LSM construction methods by radio measurements when the channel model parameters are known. 
The proposed approaches are validated by extensive numerical simulations in section IV, and finally, Section V summarises the key results and main conclusions of the paper.

\section{System model}\label{sec:System model}
As shown in Fig. \ref{F: LSM construction via CL}, we consider the LSM construction problem for cellular-connected UAV. 
The GBS is assumed to be located at the origin and the space of interest is denoted by the cubic $ \mathcal{X} = [0, W]\times[0, L] \times [0,H]$, where $H$ is the UAV flying height and it is assumed to be larger than the maximum building heights. To assist the UAV service deployment and trajectory design, we want to build a probabilistic LSM, which returns the probability for the existence of LoS link at a given location within the UAV flying plane. 
Specifically, we model the LSM $ \mathcal{M} $ as a \textbf{\textit{spatial random field}} \cite{christakos2012random, brett2003introduction}, i.e $ \mathcal{M}=\{ l(\mathbf{x})\}_{\mathbf{x} \in \mathcal{X}_h } $. Each $ l(\mathbf{x}) $ is a binary random variable representing the existence of LoS link at location $ \mathbf{x} \in \mathcal{X}_h $, and $l(\mathbf{x}) =1$ indicates that the LoS link exists, otherwise, $ l(\textbf{x}) = 0$, where  $ \mathcal{X}_h = [0,W] \times [0,L] \times H$. 

\begin{figure}[htb]
	\centering
	\includegraphics[scale=0.5]{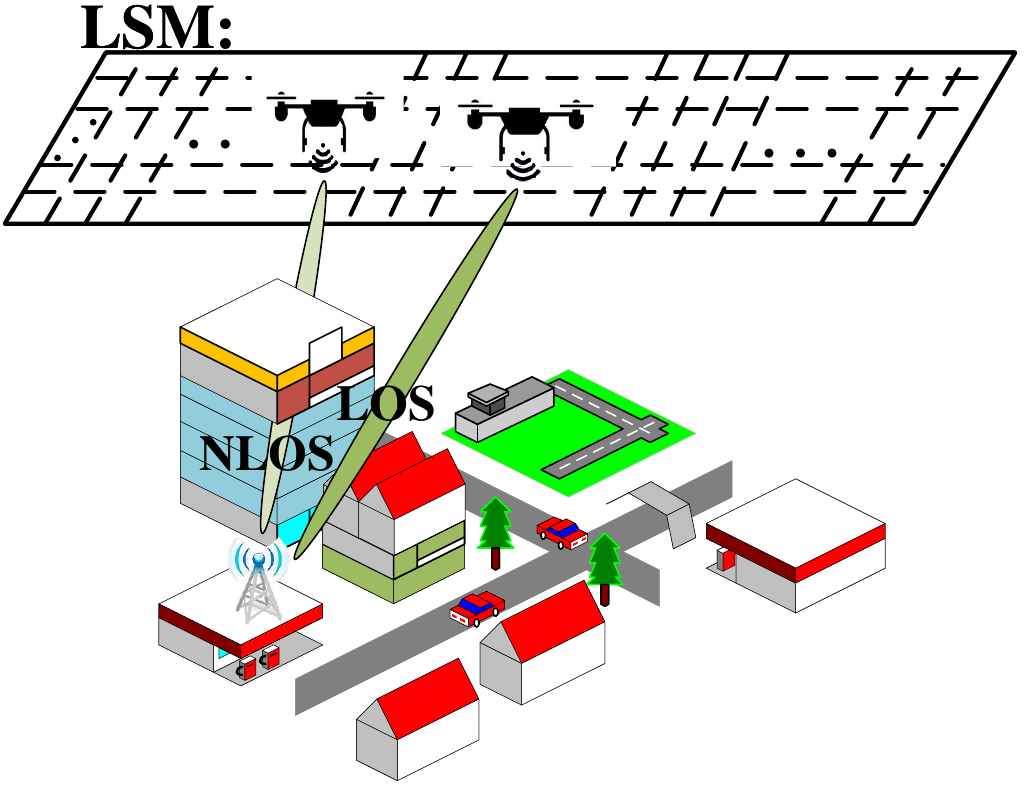}
	\caption{LSM construction for cellular-connected UAVs.}
	\label{F: LSM construction via CL}
\end{figure}

The probability of LoS link is related with the UAV elevation angle and the propagation environment, which motivates various LoS link probability models. For example, a commonly adopted model is \cite{holis2008elevation}   
\begin{align}
\Pr( \mathbf{x} ) = a - \frac{a-b}{1+(\frac{\phi(\mathbf{x})-c}{d})^e},
 \label{eq:LoSModel}
\end{align}
where  $a $, $b $, $c $, $d $, and $e $ are environment-related parameters determined from experience model and $\phi(\mathbf{x})=\arctan\left(\frac{H}{\|\mathbf{x}\|}\right)$ is the elevation angle between the UAV and the GBS. 

The existing LoS probability model can be used to obtain a prior probability distribution of $ \mathcal{M} $, denoted by $\mathcal{M}_0  $, i.e $  \mathcal{M}_0(\mathbf{x}) \triangleq \Pr(l(\mathbf{x})=1) = \Pr(\mathbf{x}) $. However, the experience model is only based on the large-scale environment information, such as the environment type (urban or rural) or the statistic building distribution, which cannot reflect the real LoS conditions of a specific environment. To tackle this problem, some specific environment information or on-site channel measurements must be used to construct the high-quality LSM. 
In this paper, we consider updating $ \mathcal{M} $ based on a limited amount of dedicated channel measurements so that its probability distribution is as close as possible to the ground-truth LSM.




\vspace{-1ex}
\subsection{Measurement model}
The measured channel information obtained by the UAV is related with the measurement location, which can be generally modeled as 
\begin{equation} \label{Radio measurements Model}
	z(\mathbf{x})=g(\mathbf{x})+n_0,
\end{equation}
where $z(\mathbf{x})$ is the measured channel information at the location $ \mathbf{x}$, the function $g(\cdot)$ describes the relation between the location and the measured results, and $n_0{\sim}N(0,\sigma^{2})$ is the measurement noise. For example, the channel gain in dB scale can be expressed as
\begin{align}
    g(\mathbf{x})=\beta_{c}+10\alpha_{c}\log_{10}(\|\mathbf x\|)+\eta_c, \label{eq:chgain}
\end{align}
where $c\in\{0,1\}$ represents the NLoS and LoS propagation environments, respectively, with $ c = 1$ corresponding to LoS, $ c = 0  $ to NLoS. $\{\alpha_c,\beta_c\}$ are the path loss parameters, satisfying  $\beta_{1} > \beta_{0} $ and $ \alpha_{1}>\alpha_{0}$, the channel fading $\eta_c$ is the random variable distributed according to $\mathcal{N}(0, \sigma_c^2)$.  

 The measured channel gain $z(\mathbf{x})$ is then a mixed Gaussian random variable as shown in Fig.~\ref{F:RSS}.  
 The conditional probability distributions are given by
\begin{align}
&p_{z|c}(\mathbf{x})\triangleq\Pr(z(\mathbf{x})|l(\mathbf{x}) = c)\nonumber\\
&=\frac{1}{\sqrt{2\pi(\sigma^2+\sigma_c^2)}}\exp\left(-\frac{(z(\mathbf{x})-\mu_c(\mathbf{x}))^2}{2(\sigma^2+\sigma_c^2)}\right), \label{eq:zxdist}
\end{align}
where $\mu_c(\mathbf{x})=\beta_{c}+10\alpha_{c}\log_{10}(\|\mathbf x\|)$. 
The distribution of $z(\mathbf{x})$ in \eqref{eq:zxdist} can be used to derive the LoS posterior probability $ \Pr(l(\mathbf{x})=c|z(\mathbf{x})) $ only measuring $ z(\mathbf{x}) $ at $\mathbf{x}$, called \textit{inverse measurement model} \cite{thrun2002probabilistic}. 
  However, to accurately construct the whole LSM, more measurements and the spatial correlation among the LoS link status  have to be considered. 

\begin{figure}[htb]
	\centering
	\includegraphics[scale=0.4]{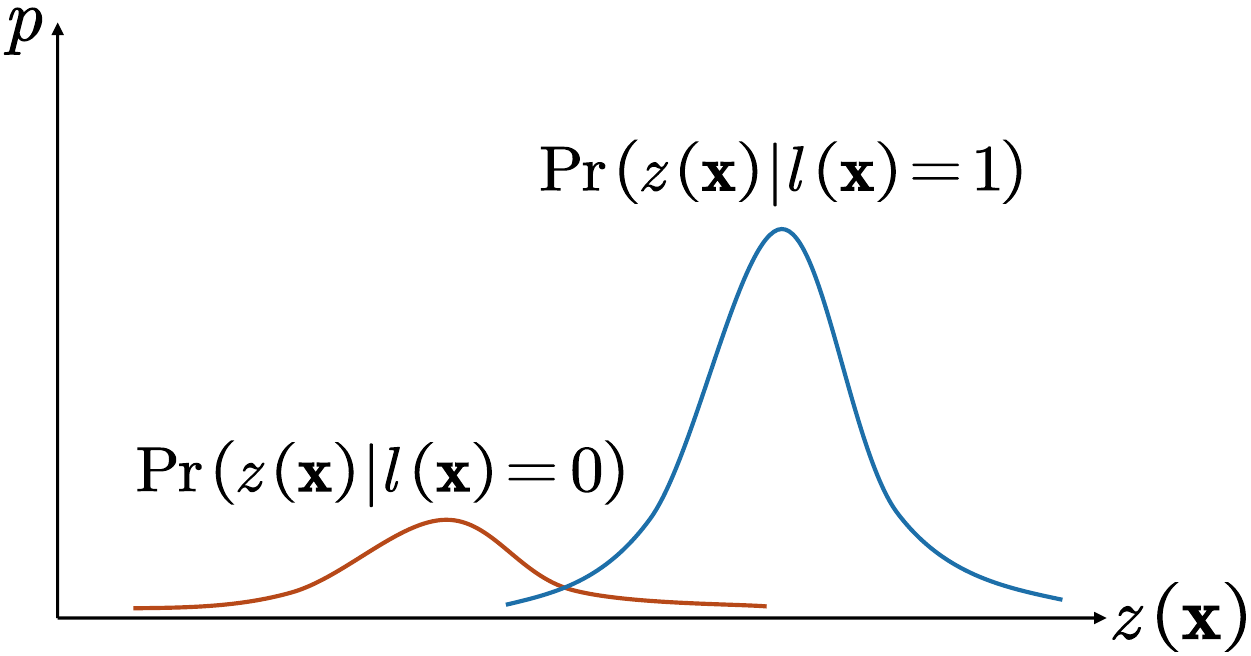}
	\caption{The distribution of channel measurements conditioned on the LoS link status.}
	\label{F:RSS}
\end{figure}



\subsection{An Overview of Binary Bayesian Filter}
This subsection aims to give a brief overview on binary Bayesian filter for sequential binary random variable update, by introducing the key notations
to be used in the sequel of this paper. The readers are referred to the classic textbook \cite{thrun2002probabilistic} for
a more comprehensive description.


Binary Bayesian filter is a useful mathematical tool for the estimation problem with binary 
state that does not change over time. 
Denote $ x $ as the binary random variable, whose state has only two cases.
When the binary state is static, the belief $ bel_t(x) $ only depends on the measurements, i.e.
	\begin{align}
		bel_t(x) = p(x|\mathbf{z}_{1:t},\mathbf{u}_{1:t}) = p(x|\mathbf{z}_{1:t}),
	\end{align}
where $ \mathbf{z}_{1:t} $ and $ \mathbf{u}_{1:t} $ are measurements and environment state, respectively. The lack of a time index for the state $ x $ reflects that the binary state is static.
To avoid the truncation problems that arise for probabilities close to $ 0 $ or $ 1 $, the binary Bayesian filter is commonly implemented by the form of \textit{log odds ratio}.
The log odds $ l_t $ is the logarithm of this expression 
\begin{align}
	l_t := \ln \frac{p(x|\mathbf{z}_{1:t})}{1-p(x|\mathbf{z}_{1:t})}.
\end{align}
The corresponding basic update algorithm is summarized as Algorithm~\ref{alg1}.
\begin{algorithm}
	\caption{Binary Bayesian Filter}
	\label{alg1}
	\begin{algorithmic}[1]
		\STATE{\textbf{Input} $ l_{t-1},z_t $ }
		\STATE{$ l_t = l_{t-1} + \ln \frac{p(x|z_t)}{1-p(x|z_t)} - \ln \frac{ p(x)}{1-p(x) } $}
		\STATE{\textbf{return} $ l_t $}
	\end{algorithmic}
\end{algorithm}

This binary Bayesian filter in Algorithm~\ref{alg1} is updated by \textit{inverse measurement model} $ p(x|z_t)$. This is because forward measurement model $ p(z_t|x) $ would be more complex to implement than inverse measurement model $ p(x|z_t) $ in general due to the binary state having only two cases, while the state space of measurements is huge.

\subsection{LSM Construction}
Denote by $\mathbf{X}=\{\mathbf{x}_1,\mathbf{x}_2,\cdots,\mathbf{x}_N\}$ the set of locations where dedicated channel measurements are collected. The corresponding measured channel gains are denoted by the set $\mathbf{Z}=\{z_1,\cdots,z_N\}$.  
The LSM construction problem is formulated as 
\begin{equation}
\mathcal{M}=f(\mathbf{X},\mathbf{Z},\mathcal{M}_0). \label{eq:constfun}
\end{equation}

To measure the quality of the constructed LSM, we consider the ground-truth LSM, denoted by $\mathcal{T}$, which gives the true LoS link status at arbitrary locations, i.e., $\mathcal{T}(\mathbf{x})=1$ when LoS is present at the location $ \mathbf{x}$, otherwise $\mathcal{T}(\mathbf{x})=0$. Then, the quality of the constructed LSM can be evaluated by the mean absolute error (MAE) between $\mathcal{T}$ and $\mathcal{M}$, i.e.,
\begin{align}\label{eq:MAE}
		MAE(\mathcal{T};\mathcal{M}) = \frac{1}{| \mathcal{X}_h |}\sum_{\mathbf{x} \in \mathcal{X}_h} | \mathcal{T}(\mathbf{x}) - \mathcal{M}(\mathbf{x}) | d\mathbf{x}.
\end{align}

In this paper, we consider the design of LSM construction function $ f(\cdot) $ based on the given channel measurement and binary Bayesian filter, so that the MAE between the constructed LSM  $ \mathcal{M} $ and the ground-truth $ \mathcal{T} $ is minimized.
Mathematically, the problem is formulated as
\begin{equation}\label{eq:OP1}
	\begin{aligned}
	\textbf{(P1)}:	&\min_{f} \quad  MAE(\mathcal{T};\mathcal{M})\\
 &\textnormal{s.t. \eqref{eq:constfun}-\eqref{eq:MAE}}. 
	\end{aligned}
\end{equation}

Solving \textbf{P1} is rather challenging and the reason is twofold. 
First, the ground-truth LSM $\mathcal{T}$ is not available and hence the explicit expression of the objective function $ f(.) $ in \eqref{eq:MAE} can not be computed analytically.
Second, due to the complex interplay between the environment information and communication channel measurements, how to exploit the environment semantics or radio propagation rule for the spatial correlation to enhance the LSM construction is still quite challenging.
To tackle these problems, we propose to use binary Bayesian filter for LSM construction from given measurements and spatial correlation of LoS probability to enhance LSM construction.
\section{Link State Mapping based on Radio Measurements}\label{sec:Mapping}
In this section, we consider the LSM construction based on a given set of radio measurements and propose to replace the construction function $f(\cdot)$ with a sequential probability updating algorithm.  Before observing any radio measurement, the initial distribution of LSM $ \mathcal{M}_0$ is constructed based on the prior LoS link probability function, such as \eqref{eq:LoSModel}. The prior LSM is gradually refined when more radio measurements are obtained, and $\mathcal{M}_n$ denotes the LSM at the $n$-th interval, where $\mathcal{M}_n(\mathbf{x})\triangleq \Pr(l(\mathbf{x})|\mathbf{Z}_n)$ is the estimated posterior LoS probability at location $\mathbf{x}$ based on the observations $\mathbf{Z}_n=\{z_1,z_2,...,z_n\}$, $ n \leq N $.  

\subsection{Probabilistic LSM update based on binary Bayesian filter}

Directly updating the LoS probability using binary Bayesian filter  may face the problem of 0-1 value overflow \cite{thrun2002probabilistic,edition2013bayesian}. To tackle this issue, we consider the logarithmic probability ratio $ \mathcal{L}_n (\mathbf{x})$ instead, where
\begin{align}
    \mathcal{L}_n(\mathbf{x})=\ln \frac{\mathcal{M}_n(\mathbf{x})}{1-\mathcal{M}_n(\mathbf{x})}.\label{eq:defLog}
\end{align}
Equivalently, we have $ \mathcal{M}_n(\mathbf{x}) = 1/(1+e^{-\mathcal{L}_n (\mathbf{x})})$. 
The  binary Bayesian filter calibrated for for probabilistic LSM construction is presented in Lemma~\ref{lem:BBF}. 
\begin{lemma}\label{lem:BBF}
    After receiving the $n$-th radio measurement $z_n$, the logarithmic probability ratio $ \mathcal{L}_{n} (\mathbf{x})$ can be obtained as
    \begin{align}
        \mathcal{L}_n(\mathbf{x}) &= \mathcal{L}_{n-1}(\mathbf{x}) + \ln \underbrace{\frac{\Pr(l(\mathbf{x})=1|z_n)}{1-\Pr(l(\mathbf{x})=1|z_n)}}_{k(\mathbf{x},z_n)}-\mathcal{L}_0(\mathbf{x})\label{eq:recursiveLn},
    \end{align}  
   where $\Pr(l(\mathbf{x})=1|z_n)$ is the posterior LoS probability at $\mathbf{x}$ and $k(\mathbf{x},z_n)$ is the posterior probability ratio at $\mathbf{x}$ based on the single measurement $z_n$.
\end{lemma}
\begin{proof}
Please refer to Appendix~\ref{Appendix A}.
\end{proof}

The posterior LoS probability $\Pr(l(\mathbf{x})=1|z_n)$ is also known as inverse measurement model \cite{thrun2002probabilistic}, which can be derived from channel measurement condition probability distribution in \eqref{eq:zxdist}. Specifically, at the measurement location $\mathbf{x}_n$, the posterior LoS probability is obtained by Bayesian equation as
    \begin{align}
    &\Pr(l(\mathbf{x}_n)=1|z_n)=\frac{\Pr(z_n|l(\mathbf{x}_n)=1)\Pr(l(\mathbf{x}_n)=1)}{\Pr(z_n)}\nonumber \\
        &=\frac{p_{z|1}(\mathbf{x}_n)\mathcal{M}_{0}(\mathbf{x}_n)}{p_{z|1}(\mathbf{x}_n)\mathcal{M}_{0 }(\mathbf{x}_n)+p_{z|0}(\mathbf{x}_n)(1-\mathcal{M}_{0}(\mathbf{x}_n))}, \label{eq:updateLoc}
    \end{align}
    where $p_{z|c}(\mathbf{x}),c\in\{0,1\}$ is defined in \eqref{eq:zxdist}.

For other locations $\mathbf{x}\neq \mathbf{x}_n$,  the posterior probability can be inferred from the spatial correlation of LoS probabilities. In particular, we define the spatial correlation between  $ \mathbf{x} $ and the measurement location $ \mathbf{x}_n $ by a set of probabilities $\{r_{ij}(\mathbf{x},\mathbf{x}_n)\}$, where 
\begin{align}
    r_{ij}(\mathbf{x},\mathbf{x}_n)\triangleq\Pr(l(\mathbf{x})=i|l(\mathbf{x}_n)=j), i,j\in\{0,1\}.\label{eq:spatialCorr}
\end{align}

If $\{r_{ij}(\mathbf{x},\mathbf{x}_n)\}$ is obtained, then the posterior LoS probability $\Pr(l(\mathbf{x})=1|z_n)$ for $\mathbf{x}\neq\mathbf{x}_n$ can be obtained as
    \begin{align}
     \Pr&(l(\mathbf{x})=1|z_n)= \Pr(l(\mathbf{x}_n)=1|z_n)r_{11}(\mathbf{x},\mathbf{x}_n) \nonumber\\
     &+(1-\Pr(l(\mathbf{x}_n)=1|z_n))r_{10}(\mathbf{x},\mathbf{x}_n). \label{eq:updateLocNM}
    \end{align}



Next, we consider the derivation of $\{r_{ij}(\mathbf{x},\mathbf{x}_n)\}$ based on radio propagation property and high level environmental semantics. The derivation of $ \Pr(l(\mathbf{x})\mid z_n) $ refers to Appendix \ref{Appendix B}.

\subsection{The spatial correlation of LoS link probability}
The spatial correlation between an arbitrary location $\mathbf{x}$ and the measurement location $\mathbf{x}_n$ is environment-specific, which lacks the general model. However, as shown in Fig.~\ref{F:spaCorr}, from the geometric requirement of the LoS link, {it is clear that if the LoS link at $\mathbf{x}_n$ is blocked, the locations with the same angle beyond $\mathbf{x}_n$ must also have NLoS link. Similarly, if the LoS link at $ \mathbf{x}_n $ exists, the locations with the same angle forward of $ \mathbf{x}_n $ must also have LoS link.} Besides, due to the substantial size of the building, the locations in proximity tend to have the same LoS conditions. These observations are used to characterize the  spatial correlation of LoS probabilities and hence enable the proposed binary Bayesian filter update. 

\begin{figure}[htb]
	\centering
	\includegraphics[scale=1.051]{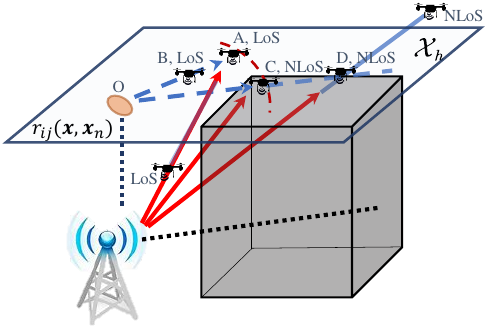}
	\captionsetup{justification=justified}
	\caption{The LSM spatial correlation in different locations.  For example, if point C is NLoS, then point D, which is at the same angle with respect to point O and behind C, must also be NLoS.
		If point A is LoS, then point B, at the same angle with respect to point O and between A and O, must also be LoS. Point O is the projection of GBS antenna on the $ \mathcal{X}_h $.}
	\label{F:spaCorr}
\end{figure}


Denote by $(r_n,\phi_n)$ the polar coordinates of the measurement location $\mathbf{x}_n$,  where $r_n$ is the distance between $\mathbf{x}_n$ with the projection of the GBS on the UAV flying plane and $\phi_n$ is the azimuth angle. We first consider the LSM update on locations with the same azimuth angle $\phi_n$. For a typical location $\mathbf{x}=(r,\phi_n)$, the LoS probability can be updated based on the following rules:
\begin{itemize}
	\item {If the communication link between location $ \mathbf{x}_n$ and the GBS is LoS, then the communication link of any location $ \mathbf{x}$ between location $ \mathbf{x}_n$ and GBS is LoS, i.e. $ r_{11} (\mathbf{x},\mathbf{x}_n ) = 1$, for any $ r < r_n$. On the other hand, for $r>r_n$, the posterior LoS probability can be updated based on binary Bayesian filter, according to \eqref{eq:updateLoc}.}
	\item{ If the communication link between location $ \mathbf{x}_n$ and the GBS is NLoS, 	then the communication link of the location  $ \mathbf{x}$ on the extended line between location  $ \mathbf{x}_n$ and the GBS towards location $ \mathbf{x}_n$ is NLoS, i.e. $ r_{00} (\mathbf{x},\mathbf{x}_n ) = 1$, for any $ r > r_n$. On the other hand, for $r<r_n$,  the posterior LoS probability can be updated based on binary Bayesian filter, according to \eqref{eq:updateLoc}}
\end{itemize}

Based on the above radio-propagation rules, we can derive the spatial correlation of LoS probability on locations $\mathbf{x}=(r,\phi_n)$ and $\mathbf{x}_n=(r_n,\phi_n)$ with the same azimuth angle $ \phi_n $ as 
\begin{equation}\label{eq:CorrelationOnAngle}
	\begin{aligned}
		r_{ij} (\mathbf{x},\mathbf{x}_n )= \begin{cases}
			1, & i = j = 1, r \leq r_n\\
			\frac{\Pr(l(\mathbf{x})=0)}{\Pr(l(\mathbf{x}_n)=0)},  & i = j = 0, r < r_n \\
			\frac{\Pr(l(\mathbf{x})=1)}{\Pr(l(\mathbf{x}_n)=1)}, & i = j = 1, r > r_n \\
			1, & i = j = 0, r \geq r_n\\
		\end{cases}
	\end{aligned}
\end{equation}

Next, we consider the spatial correlation of LoS probability for locations on the same distance, but at different azimuth directions, i.e., for $\mathbf{x}=(r_n,\phi)$, where $\phi\neq\phi_n$. Based on the radio propagation property shown in Fig.~\ref{F:spaCorr}, the location $\mathbf{x}$ tends to have the same link status with $\mathbf{x}_n$ if they are in proximity, i.e., with small azimuth angle difference, denoted by $\Delta\phi=|\phi-\phi_n|$. However, the proper statistic model for the spatial correlation of LoS link is not available yet. To tackle this problem, we first need to introduce a parameter that captures the correlation of $l(\mathbf{x})$ and $l(\mathbf{x}_n)$, and then choose the proper correlation function which matches real communication environment. 

Since $l(\mathbf{x})$ and $l(\mathbf{x}_n)$ are binary random variables, we use the phi coefficient  to quantize the level of correlation between them, which is equivalent to the Pearson correlation coefficient reduced to binary case \cite{PhiCoef}. Specifically, denote  the joint probability of   $l(\mathbf{x})$ and $l(\mathbf{x}_n)$ as 
\begin{align}
    p_{ij}=\Pr(l(\mathbf{x})=i,l(\mathbf{x}_n)=j), i,j\in\{0,1\}.
\end{align}
Then, the phi coefficient is defined as
\begin{equation}
	\begin{aligned}
		\rho = \frac{p_{11}p_{00}-p_{10}p_{01}}{\sqrt{(p_{11}+p_{10})(p_{01}+p_{00})(p_{01}+p_{11})(p_{00}+p_{10})}},
	\end{aligned}\label{eq:phi}
\end{equation}
where $\rho\in [0,1]$ indicates the level of correlation. For example, when $l(\mathbf{x})$ and $l(\mathbf{x}_n)$ are independent, we have  $\rho=0$. When $l(\mathbf{x})$ and $l(\mathbf{x}_n)$ are highly correlated, i.e., with $l(\mathbf{x})=l(\mathbf{x}_n)$, we have $\rho=1$. 

Since the obstacles are usually of limited size, the correlation between  $l(\mathbf{x})$ and $l(\mathbf{x}_n)$ should decay quickly with the azimuth angle difference $\Delta\phi$ and it becomes 0 when $\Delta\phi=\pi$. To exploit such environmental semantics, we introduce an exponential LoS correlation model as
\begin{equation}
	\rho(\Delta\phi)  = 1 - \exp(\beta(1-\frac{\pi}{\Delta \phi})),\  \Delta\phi\in [0,\pi], \label{eq:corModel}
\end{equation}
where $\beta$ is an adjustable parameter. If the blockages in the propagation environment are small in size and relatively far away from the GBS, we can choose a small value of $\beta$ to indicate the quickly de-correlated LoS condition. On the other hand, the blockages are close to the GBS or large in size, $\beta$ can be chosen to be large. Without environment information, we suggest to choose $\beta=1$. Fig.~\ref{F:rho-beta} shows the spatial correlation curves with different $ \beta $.
\begin{figure}
	    \centering
	\includegraphics[scale=0.645]{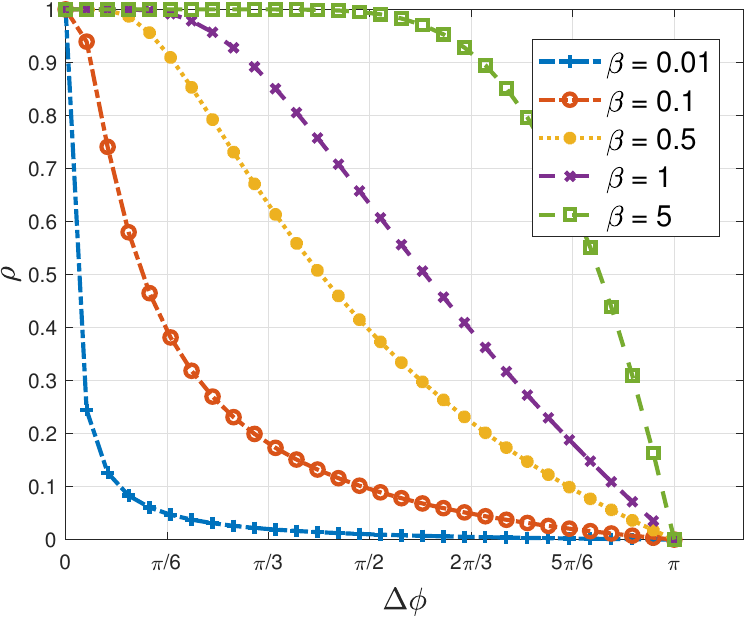}
	\caption{The spatial correlation curve $ \rho - \Delta \phi $ corresponding to  different $ \beta$.}
	\label{F:rho-beta}
\end{figure}

With the established correlation model in \eqref{eq:corModel}, we can calculate the phi coefficient between the location $\mathbf{x}=(r_n,\phi)$ and the measurement location $\mathbf{x}_n=(r_n, \phi_n)$ as $\rho(\Delta\phi)$. Then, based on the definition of phi coefficient in \eqref{eq:phi}, we can calculate the spatial correlation of LoS probability on locations $ \mathbf{x} $ and $ \mathbf{x}_n $ with the same distance $ r_n $ as
	\begin{align}\label{eq:Correlationondistance}
		&r_{ij}(\mathbf{x},\mathbf{x}_n )= \Pr(l(\mathbf{x})=i|l(\mathbf{x}_n)=j) \nonumber\\
  &= \begin{cases}
			\Pr(l(\mathbf{x})=i) +  \frac{\rho}{T} \Pr(l(\mathbf{x}_n=1-i)),  i = j,  \phi \neq \phi_n \\
			\Pr(l(\mathbf{x})=i) -  \frac{\rho}{T} \Pr(l(\mathbf{x}_n=i)),   i \neq j, \phi \neq \phi_n \\
		\end{cases}
	\end{align}
where 
\begin{equation}\label{eq:definitionT}
	\begin{aligned}
		T \triangleq \sqrt{\frac{\Pr(l(\mathbf{x}_n)=1)\Pr(l(\mathbf{x}_n)=0) }{\Pr(l(\mathbf{x})=1)\Pr(l(\mathbf{x})=0) }}.
	\end{aligned}
\end{equation}
Under the setting of LoS prior probability in \eqref{eq:LoSModel}, we have $ T = 1$ due to the equation $ \mathcal{M}_0(\mathbf{x}) = \mathcal{M}_0(\mathbf{x}_n) $, which holds for the LoS prior probability at the same distance $ r = r_n $.
Note that if there is no prior information on these two LoS random variables, it reduces to $ r_{11} = r_{00} = \frac{1+\rho}{2} $ and $ r_{10} = r_{01} = \frac{1-\rho}{2} $.
The derivation details of $ r_{ij} $, i.e. equation \eqref{eq:CorrelationOnAngle}, \eqref{eq:Correlationondistance}, are shown in the Appendix \ref{Appendix C}.

\vspace{-1ex}
\subsection{LSM Construction Algorithm}
With the spatial correlation models derived in the preceding subsection, we can enhance the prior LSM based on the collected radio measurements $\mathbf{X}$. For notational convenience, we consider the polar grid representation of LSM as shown in Fig.~\ref{F:polarGrid}. The LSM construction is completed in two steps:
\begin{itemize}
\item{\emph{Step 1}: For each measurement collected, the LoS probability on the same azimuth direction is updated based on spatial correlation model in \eqref{eq:CorrelationOnAngle}. }
\item {\emph{Step 2}: The directions without radio measurement are updated by the adjacent directions, based on the spatial correlation model in \eqref{eq:Correlationondistance}}. 
\end{itemize}
\begin{figure}[htb]
    \centering
    \includegraphics[scale=1.1]{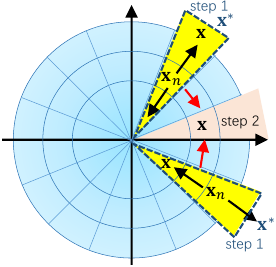}
    \caption{The LSM represented in Polar Grid.}
    \label{F:polarGrid}
\end{figure}

For step 1, considering a typical measurement $z_n$ at $\mathbf{x}_n=(r_n,\phi_n)\in\mathbf{X}$, we first compute the posterior LoS probability $ \Pr(l(\mathbf{x}_n)|z_n) $  based on \eqref{eq:updateLoc}, and  update logarithmic probability ratio $\mathcal{L}_n(\mathbf{x}_n)$ based on Lemma~\ref{lem:BBF}. Next, for the locations with the same azimuth angle, i.e., $\mathbf{x}=(r,\phi_n)$, the posterior LoS probability $\Pr(l(\mathbf{x})=1|z_n)$  can be computed by substituting the correlation function in \eqref{eq:CorrelationOnAngle} into \eqref{eq:updateLocNM}. Then, the  logarithmic probability ratio at $\mathbf{x}$ can be updated based on Lemma~\ref{lem:BBF}. Specifically, the posterior probability ratio used in \eqref{eq:recursiveLn} for LSM update is given by
\begin{align}\label{eq:invmeasmodel1}
k(\mathbf{x},z_n)=\begin{cases}
    \frac{k(\mathbf{x}_n,z_n)(1-\mathcal{M}_0(\mathbf{x}_n))-\mathcal{M}_0(\mathbf{x}_n)+\mathcal{M}_0(\mathbf{x})}{1-\mathcal{M}_0(\mathbf{x})} & r<r_n\\  
    \frac{k(\mathbf{x}_n,z_n)\mathcal{M}_0(\mathbf{x})}{k(\mathbf{x}_n,z_n) ( \mathcal{M}_0(\mathbf{x}_n)-\mathcal{M}_0(\mathbf{x}) ) + \mathcal{M}_0(\mathbf{x}_n)}, & r\geq r_n
\end{cases}.
\end{align}

Step 2 starts after all the radio measurements have been considered. Denote the directions with radio measurements as $\Phi=\{\phi_n,n=1,...,N\}$, which have been updated in step 1. Then, the directions without radio measurements are given by the set $\bar\Phi=[0,2\pi)\setminus\Phi$. The LoS probability on those directions is updated based on the spatial correlation model in \eqref{eq:Correlationondistance}. To reduce the complexity, we consider an angle threshold for LoS link status, denoted by $\phi_{\mathrm{th}}$, beyond which the LoS correlation is assumed to be negligible. Consider a typical direction $\phi\in\bar{\Phi}$, its angle correlation region is defined as $\Gamma(\phi)=(\phi-\phi_{\mathrm{th}}, \phi+\phi_{\mathrm{th}})$. If there is no measurement within the correlation region, i.e., $\Gamma(\phi)\bigcap \Phi=\emptyset$, the direction $\phi$ is not updated. If there are multiple directions in $\Phi$ within the correlation region, the LSM is updated based on the nearest direction. Denote by $\phi^*=\arg\min_{\phi'\in\Phi}|\phi-\phi'|$ the nearest direction with measurement. Then, the posterior LoS probability at $\mathbf{x}=(r,\phi)$ is updated based on the posterior LoS probability at $\mathbf{x}^*=(r,\phi^*)$ using Bayesian principle.  Mathematically, we have
\begin{align}
     \Pr(l(\mathbf{x})=1|\mathbf{Z})&= \mathcal{M}_N(\mathbf{x}^*)r_{11}(\mathbf{x},\mathbf{x}^*) \nonumber \\
     &+(1-\mathcal{M}_N(\mathbf{x}^*))r_{10}(\mathbf{x},\mathbf{x}^*),
\end{align}
where $\mathcal{M}_N(\mathbf{x}^*)$ is the posterior LoS probability at $\mathbf{x}^*$ after all the measurements have been processed, $r_{ij}(\mathbf{x},\mathbf{x}^*)$ is the correlation between these two directions given by \eqref{eq:Correlationondistance}, and $\Pr(l(\mathbf{x})=1|\mathbf{Z})$ is the posterior LoS probability at $\mathbf{x}$ after considering the measurements on the nearby directions.
Specifically, the posterior probability for LSM update at the directions without radio measurement is given by
\begin{equation}
	\begin{aligned}\label{eq:updateLocNMA}
		\Pr(l(\mathbf{x})|\mathbf{Z}) = \mathcal{M}_0(\mathbf{x}) + \frac{\rho}{T}(\mathcal{M}_N(\mathbf{x}^*)-\mathcal{M}_0(\mathbf{x}^*)), 
	\end{aligned}
\end{equation}
where $ T $ is defined in \eqref{eq:definitionT}. Note that \eqref{eq:updateLocNMA} can be further reduced under the following special case:

\begin{itemize}
	\item When the location at $ \mathbf{x}$ and $ \mathbf{x}^* $ are highly correlated, i.e., $ \rho=1 $, we have $ \Pr(l(\mathbf{x})|\mathbf{Z}) = \mathcal{M}_N(\mathbf{x}^*) $, which implies that the LoS probability at $ \mathbf{x} $ is replaced with the more accurate measurement results at $\mathbf{x}^*$.    
	\item When the location at $ \mathbf{x} $ and $ \mathbf{x}^*$ are totally uncorrelated, i.e., $ \rho=0 $, we have $ \Pr(l(\mathbf{x})|\mathbf{Z}) = \mathcal{M}_0(\mathbf{x}) $, which implies that the LoS probability at $ \mathbf{x} $ is not updated. 
\end{itemize}
When we use the LoS prior probability model in \eqref{eq:LoSModel}, $ T = 1 $. 
the posterior LoS probability is then denoted as the update of LSM distribution, i.e. $ \mathcal{M}_N(\mathbf{x}) \triangleq \Pr(l(\mathbf{x})|\mathbf{Z}) $. The design of LSM construction function $ f(.)$ can be summarized as Algorithm \ref{alg2}.
 To construct the LSM, the first step is to perform a sequential Bayesian update using the spatial correlation of distance in \eqref{eq:CorrelationOnAngle}, which requires an exhaustive search of the measurements with a complexity overhead of $ \mathcal{O}(N) $. And then, the directions without radio measurement are updated by the adjacent directions based on the spatial correlation of angle in \eqref{eq:Correlationondistance}, the worst complexity of this step is $  \mathcal{O}(| \mathcal{X}_h| N) $. Therefore, the algorithm complexity grows linearly with the number of measurements $ N $ and the total number of grids $ | \mathcal{X}_h| $ split by $ \mathcal{X}_h $, respectively.

	\begin{algorithm}[t]
		\caption{The design of probabilistic LSM construction function }
		\label{alg2}
		\begin{algorithmic}[1]
			\STATE{\textbf{Input:} Prior LoS probability $ \mathcal{M}_0(\mathbf{x}) \triangleq  \Pr(l(\mathbf{x})=1)$, $ \mathbf{X} $, $ \mathbf{Z} $, The projection of BS $ \mathbf{x}_{bs}$.}
			\STATE{ Denote $ (0,0,H)$, $ (x_n,y_n,H)$ as the projection of GBS and radio measurement location at $ n$-th time slot, respectively. Denote any location $ \mathbf{x} = (x,y,H) \in \mathcal{X}_h$.}
			\STATE{\textbf{For} $ n = 1:N$}
			\STATE{ $r_n=\sqrt{(x_n)^2+(y_n)^2},\theta_n=\arctan2(y_n,x_n) $}
			\STATE{ $r = \sqrt{(x)^2+(y)^2},\theta=\arctan2(y,x)  $}
			\STATE{For any location $ \mathbf{x} $ in the direction with measurement  } 
			\STATE{\quad \quad Compute $ \Pr(l(\textbf{x})|z_n) $ by $ \eqref{eq:updateLoc}-\eqref{eq:CorrelationOnAngle} $ }
			\STATE{ \quad \quad Updating $ \mathcal{L}_n(\mathbf{x}) $  by binary Bayesian filter $ \mathcal{L}_n(\textbf{x}) = $  \quad  $ \mathcal{L}_{n-1}(\textbf{x}) + \ln \frac{\Pr(l(\textbf{x})|z_n)}{1-\Pr(l(\textbf{x})|z_n)}-\mathcal{L}_0(\textbf{x})$}
			\STATE{\textbf{end}}
			\STATE{Compute $ \mathcal{M}_N(\mathbf{x}) = 1/(1+e^{-\mathcal{L}_N(\mathbf{x})}) $ for the updated location}
			\STATE{For any location $ \mathbf{x} $ on the direction without measurement}
			\STATE{\quad \quad Find the nearest direction with measurement $\phi^*=\arg\min_{\phi'\in\Phi}|\phi-\phi'|$, for all direction $ \phi \in \Phi $ of measurement $\mathbf{Z}$}
			\STATE{\quad \quad \textbf{If} $ \phi^* = \emptyset $}
			\STATE{\quad \quad \quad \quad \textbf{Break}}
			\STATE{\quad \quad \textbf{Else}}
				\STATE{\quad \quad \quad \quad Denote by $ \mathbf{x}^* = (r,\phi^*) $ the location with the smallest angle distance with $ \mathbf{x} $}
			\STATE{\quad \quad \quad \quad Compute the posterior probability $ \mathcal{M}_N(\mathbf{x}) = \Pr(l(\mathbf{x})|\mathbf{Z})$ by \eqref{eq:updateLocNMA}}
			\STATE{\quad \quad \textbf{end}}
			


			\STATE{\textbf{Output}: The constructed LSM $ \mathcal{M}$ and its distribution $ \mathcal{M}_N $ }
		\end{algorithmic}
	\end{algorithm}
\vspace{-1.5ex}
\section{Numeric Results}

\begin{figure}[h]
	\centering
	
	\begin{subfigure}[b]{0.24\textwidth}
		\centering 
		\includegraphics[scale=0.334]{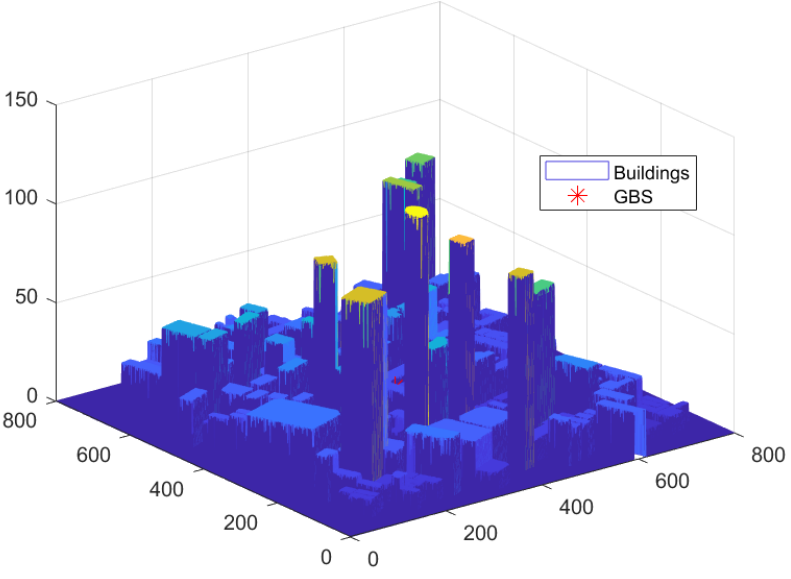}
		\caption{Simulated 3D urban area}
		\label{F: Buildings}
	\end{subfigure}
		\begin{subfigure}[b]{0.24\textwidth}
		\centering 
		\includegraphics[width=0.96\columnwidth]{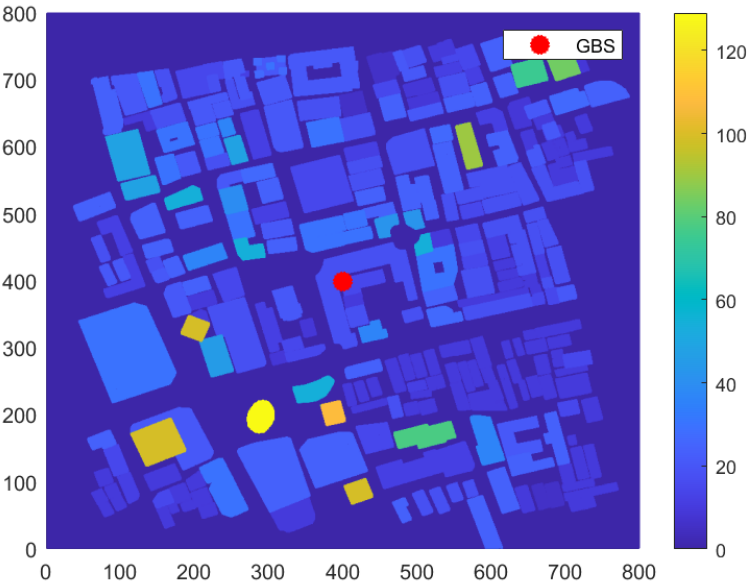}
		\caption{2D sampled location}
		\label{F:TopviewBuilding}
	\end{subfigure}
	\hfill
	\caption{The simulated urban area.}
	\label{F:BuildingsandTopviewBuilding}
\end{figure}

\vspace{-1ex}
In this section, we evaluate the proposed algorithms in an urban area of $ 800 $m $ \times 800 $m, as shown in Fig. \ref{F:BuildingsandTopviewBuilding}, 
 where a GBS communicates with a UAV. We consider the probabilistic LSM construction at the fixed plane $ \mathcal{X}_h$, the height of UAV is set as $ H = 129$m. 
The dedicated channel model parameters are set by 3GPP TR36.777 UMa-AV model \cite{3GPP}. The channel gain in dB can be written as follows. 
\begin{equation}
		\begin{aligned}
		&g(\mathbf{x}) =\\
		 &\begin{cases}
			&-28 - 22 \log_{10} (\| \mathbf{x} \|) - 20 \log_{10} (f_c) \\ &
			 + \eta_{1},  \ 
			\rm{if} \  LoS, \\
			&17.5 - 20 \log_{10} (\frac{40 \pi f_c}{3}) - (46 - 7\log_{10}(H))\log_{10}(\| \mathbf{x} \|) \\
			 & + \eta_{0}, \ \rm{if} \ NLoS,
		\end{cases}
	\end{aligned}
\end{equation}
where $ \alpha_1 = - 2.2 $, $ \alpha_{0} = - 4.6 + 0.7\log_{10}(H) $, $ \beta_{1} = -28 - 20 \log_{10}(f_c)  $, $ \beta_{0} = 17.5 - 20 \log_{10} (\frac{40 \pi f_c}{3})  $, $ \eta_{1} \sim \mathcal{N}(0,\sigma_1^2) $, $ \eta_{0} \sim \mathcal{N}(0,\sigma_0^2) $.
The prior probability of LSM $ \Pr(l(\mathbf{x})),\mathbf{x} \in \mathcal{X}_h$ is set by \eqref{eq:LoSModel} \cite{holis2008elevation}. The prior LSM and ground-truth LSM are shown in Fig.~\ref{F:realLSM and priorLSM}. 
Specifically, the empirical parameter settings are: $ a = 120, b = c = 0, d = 24.3, e = 1.229 $.
 The grid size of simulated map $ \mathcal{X}_h$  is set as $ 1 $m$\times 1 $m . The GBS is placed at $ (400,400) $m. The main simulation parameters settings about the simulated environment are shown in Table \ref{tab:Para1}.

\begin{figure}[h]
	\centering    
	\begin{subfigure}[b]{0.24\textwidth}
		\centering 
		\includegraphics[scale=0.35]{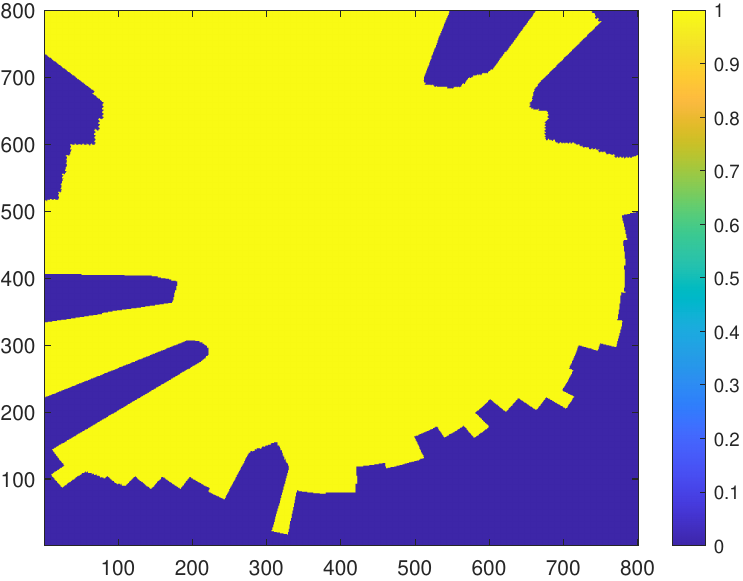}
		\caption{Ground-Truth LSM}
		\label{F:real LSM}
	\end{subfigure}
	\hfill
	\begin{subfigure}[b]{0.24\textwidth}
		\centering 
		\includegraphics[scale=0.35]{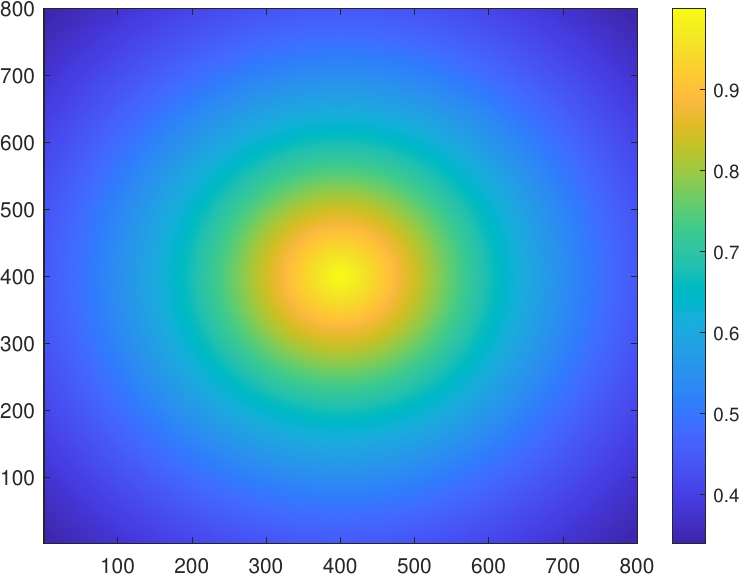}
		\caption{Prior LSM}
		\label{F:prior LSM}
	\end{subfigure}
	\hfill
	\caption{The Ground-Truth LSM and Prior LSM by our simulation environment and parameter setting.}
	\label{F:realLSM and priorLSM}
\end{figure}

\begin{table}[h] 
	
	\caption{Main Simulation Parameters}
	\centering
	\begin{tabular}{l|c}
		\hline
		\textbf{Parameter \& notations} & \textbf{Value}  \\
		\hline
		BS antenna height $ H_{bs}$ & $ 15.00 $m \\
		\hline
		Communication frequency $ f_c $ & $ 28$GHz \\
		\hline
		UAV flight height $H$ & $129$m\\
		\hline
		The grids size $\delta_x = \delta_y $ & $1m $ \\
		\hline
		The path loss exponent $\alpha_{1}$ in LoS & $-2.20$\\
		\hline
		The path loss offset $\beta_{1}$ in LoS & $-56.9431$\\
		\hline
		The variance of shadowing effects $ \eta_{1}$ in LoS & $3.9221$\\
		\hline
		The path loss exponent $\alpha_{0}$ in NLoS & $-3.12$\\
		\hline
		The path loss offset $\beta_{0}$ in NLoS & $-43.8849$\\
		\hline
		The variance of shadowing effects $ \eta_{0}$ in  NLoS & $ 6.25 $ \\
		\hline
		The variance of radio measurements $ \delta_{t}$ & $0$\\
		
		
		\hline
		
	\end{tabular}
	
	\label{tab:Para1}
\end{table}

To evaluate the performance of our algorithm, we set the following benchmark schemes:
\begin{itemize}
	\item \textit{Prior LSM}: directly given by the prior distribution in \eqref{eq:LoSModel} without any measurement.
	\item \textit{The baseline LSM by $ K $-nearest neighbours}:  directly interpolated by the LoS likelihood log ratio at the measured location. 
	$ K $-nearest neighbours (KNN) can interpolated each predicted location point based on the LoS likelihood logarithmic ratio of the $ K $ nearest neighbour measurements that have the smallest Euclidean distance from the location of the measurements.

	\begin{equation}\label{eq:KNNinterpolation1}
		\mathcal{M}_{KNN}(\mathbf{x}) = \frac{\sum_{\mathbf{x}_n \in \mathcal{N}_1(\mathbf{x})} w_n(\mathbf{x}) \Pr(l(\mathbf{x}_n)|z_n)}{\sum_{n \in \mathcal{N}_1(x)} w_n(\mathbf{x})}
	\end{equation}
	where $ \mathcal{N}_1(\mathbf{x}) $ defined in \eqref{eq:KNNinterpolation1} is the location  set of the $ K $-nearest
	neighbor measurement samples of $ \mathbf{x} $, $ \mathcal{N}_1(\mathbf{x}) \subset X $. 
	 
	
	\item \textit{The proposed LSM only with the spatial correlation of distance}: after updating the LSM on the direction with measurement, other locations without updating are interpolated by KNN using the posterior LoS likelihood log ratio $ \mathcal{L}_N(\mathbf{x}) $ of the direction with the measurements.
	Thus, for the location that have not been updated by $ \mathbf{Z}$, we use spatial interpolation, such as KNN, to update the probabilistic LSM. 
	We consider the direction of LSM that has been updated by $ \mathbf{Z}$ and resample the posterior log-probability ratio of the LSM in each direction with a smaller spacing of $ \delta_d $ than the measurements space interval. Denote $ \mathcal{S} $ as the resample location set.
	For any location $ \mathbf{x} $ without updating, KNN interpolation can be written as 
	\begin{equation}\label{eq:KNNinterpolation2}
		\mathcal{L}_N(\mathbf{x}) = \frac{\sum_{x_i \in \mathcal{N}_2(\mathbf{x}) } w_i(\mathbf{x}) \mathcal{L}_N(\mathbf{x}_i)}{\sum_{x_i \in \mathcal{N}_2(\mathbf{x})} w_i(\mathbf{x})},
	\end{equation}
	where $ \mathcal{N}_2(\mathbf{x}) $ defined in \eqref{eq:KNNinterpolation2} is the location  set of the $ K $-nearest
	neighbor resample samples of $ \mathbf{x} $ in resample location set $ \mathcal{S} $, $ \mathcal{N}_2(\mathbf{x}) \subset S $. 
	Then, the LSM $ \mathcal{M}_N(\mathbf{x}) $ can be calculated by $ \mathcal{M}_N(\mathbf{x}) = \frac{1}{1+\exp(-\mathcal{L}_N(\mathbf{x}))}$. 
	In our simulation, $ K $ is set to $ 5 $.
	
	
\end{itemize}

As shown in Fig~\ref{F:MAE-number}, we first evaluate the relationship between the MAE of LSM and the number of measurements, where the radio measurements are collected from the random location sample points on each discrete direction.  
The angle interval of the sample is $ \delta_{\phi} = \frac{\pi}{36} $, so there are $ M = 72$ discrete directions. We randomly sample the same number of measurements on each discrete direction.
It can be observed that the constructed map quality gets better as the number of measurements increases.

\begin{figure}[t]
	\centering
	\includegraphics[scale=0.6]{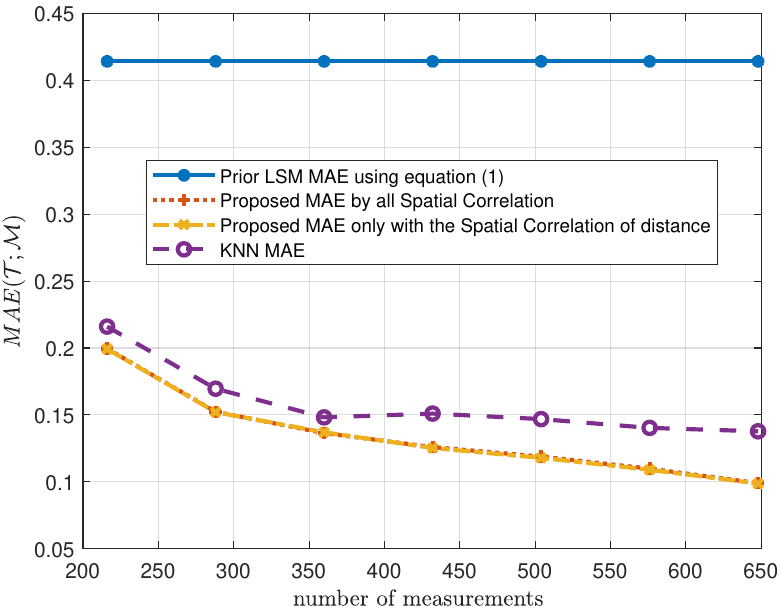}
	\caption{The performance of LSM construction algorithm $ MAE(\mathcal{T};\mathcal{M}) $ - The number of measurement $ N $.
	}
	\label{F:MAE-number}
\end{figure}

 In Fig~\ref{F:lnk-sigma0}, we evaluate the relationship between measurement quality $ \ln k(\mathbf{x}_n,z_n) $ and shadowing variance $ \sigma_0^2 $ on different measurement locations. 
We conduct extensive numerical simulations at each measurement location with different distances to GBS, with half of the simulations set to LoS conditions and the other half to NLoS conditions.  Denote $ r $ as the distance to the GBS at the measurement location.
	

\begin{figure}[t]
	\centering
    \includegraphics[scale=0.6]{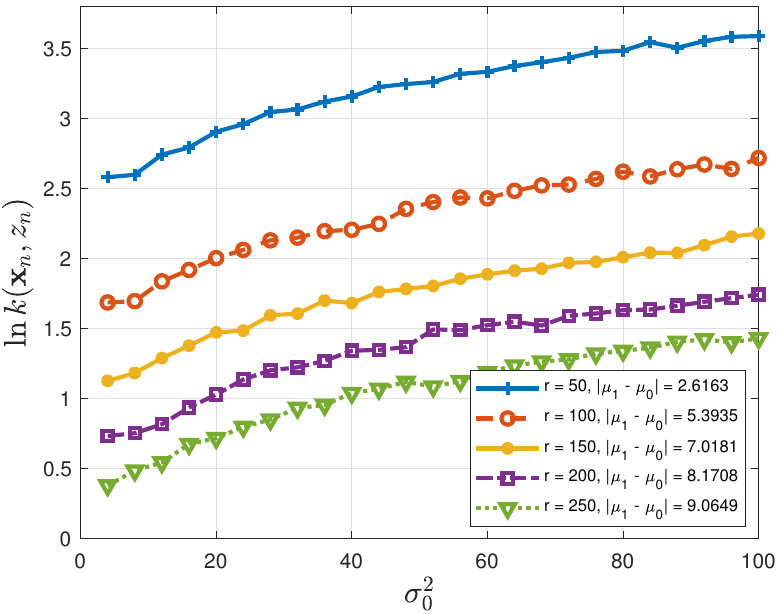}
	\caption{The measurement coefficient quality $ k(\mathbf{x}_n,z_n) $ - The shadowing variance of NLoS $ \sigma_0^2 $.
	}
	\label{F:lnk-sigma0}
\end{figure}

Fig.~\ref{F:MAE-sigma0} presents the line chart of map construction quality MAE  versus $ \sigma_0^2 $, comparing the proposed algorithm with the benchmark. Parameter settings other than $ \sigma_0^2 $ remain unchanged, as shown in Table~\ref{tab:Para1}. From the line plot of MAE in the Fig.~\ref{F:MAE-sigma0}, we can see that as the measurements quality improves, i.e. the shadowing variance $ \sigma_0^2 $ of NLoS gets smaller, LSM construction quality improves.

\begin{figure}[t]
	\centering
	\includegraphics[scale=0.6]{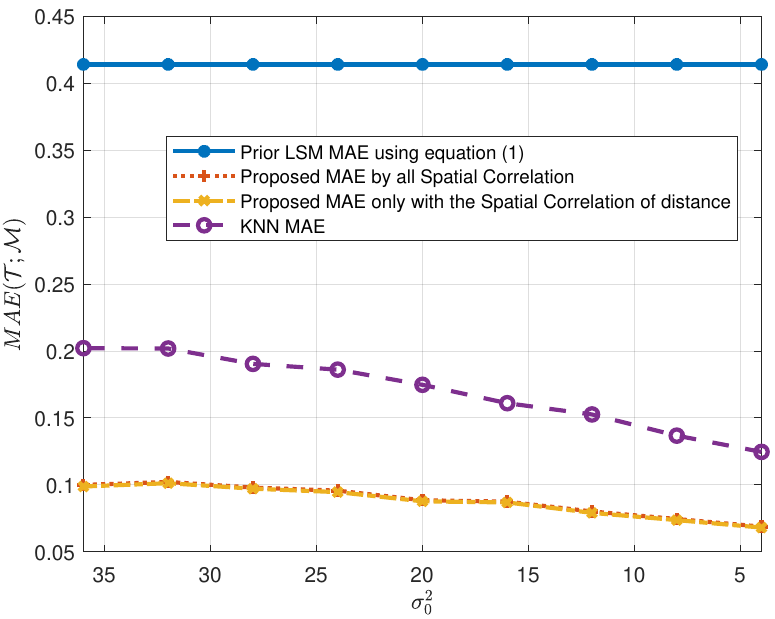}
	\caption{The performance of LSM construction algorithm $ MAE(\mathcal{T};\mathcal{M}) $ - The shadowing variance of NLoS $ \sigma_0^2 $, \\
		parameter setting: $ \beta = 1, \phi_{th} = \frac{\pi}{9}, \delta_{\phi} = \frac{\pi}{36}, \delta_{D} = 50 $m$, \delta_{d} = 1 $ m.  }
	\label{F:MAE-sigma0}
\end{figure}

\begin{figure*}[t]
	\centering    
	\begin{subfigure}[b]{0.325\textwidth}
		\centering 
		\includegraphics[scale=0.45]{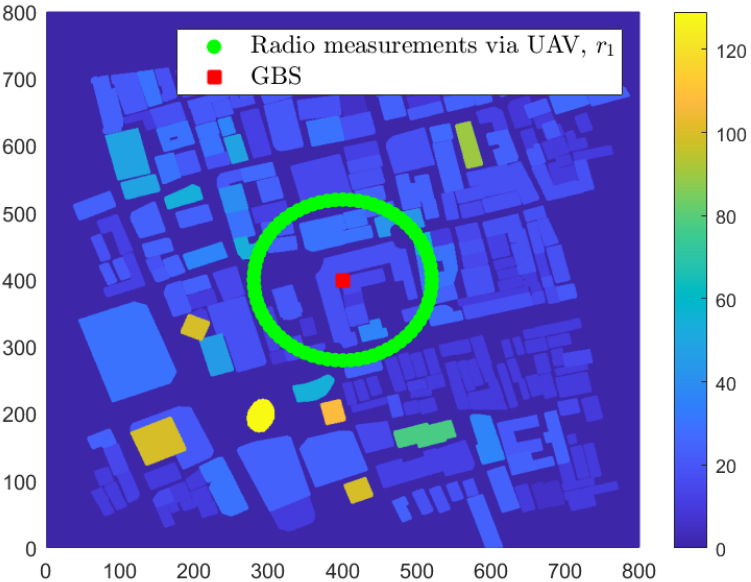}
		\caption{The sampled location by UAV from top view, $ r_1 $.}
		\label{F:sampled location,$ r_1 $}
	\end{subfigure}
	\hfill
	\begin{subfigure}[b]{0.325\textwidth}
		\centering 
		\includegraphics[scale=0.45]{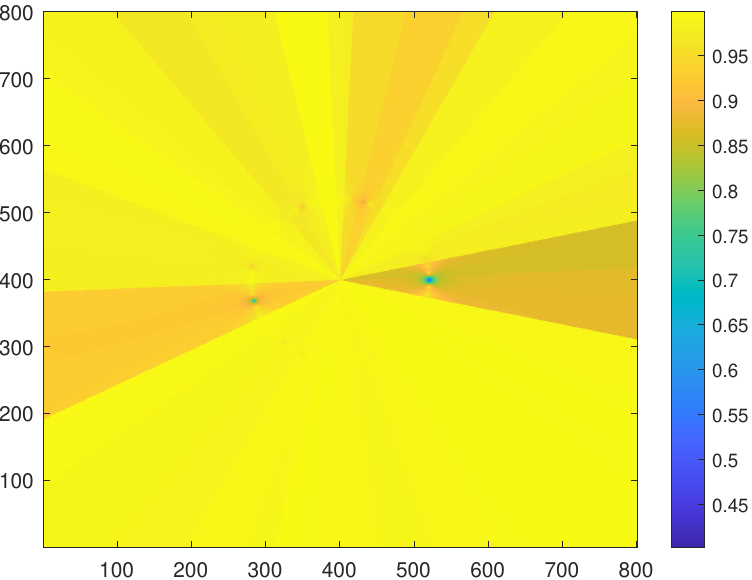}
		\caption{The baseline LSM by $ K $-nearest neighbours, $ r_1 $.}
		\label{F:KNNbaseline LSM,$ r_1 $}
	\end{subfigure}
	\hfill
	\begin{subfigure}[b]{0.325\textwidth}
		\centering 
		\includegraphics[scale=0.45]{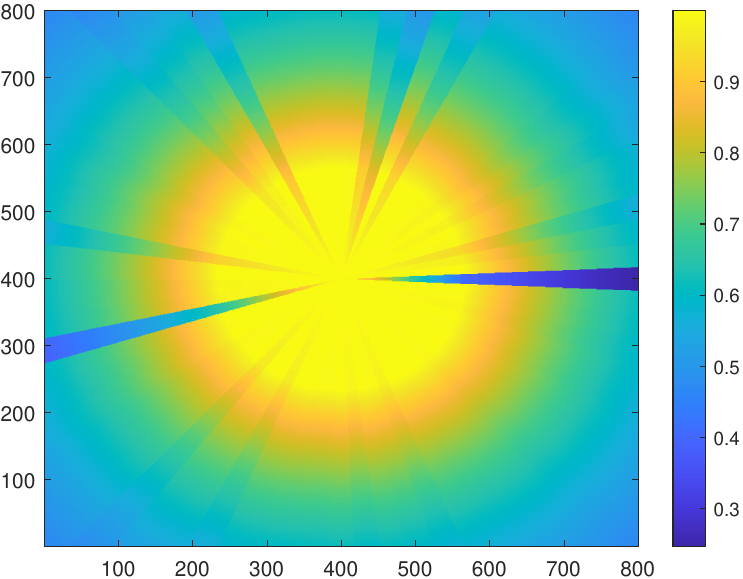}
		\caption{The Proposed LSM by all spatial correlation, $ r_1 $.}
		\label{F:The Proposed LSM by the Spatial Correlation interpolation,$ r_1 $}
	\end{subfigure}
	\hfill
	\begin{subfigure}[b]{0.325\textwidth}
		\centering 
		\includegraphics[scale=0.45]{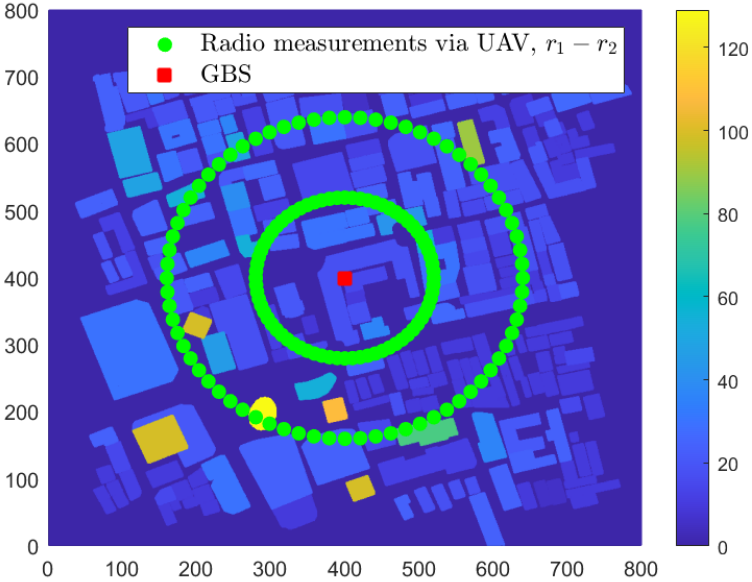}
		\caption{The sampled location by UAV from top view, $ r_1 - r_2 $.}
		\label{F:sampled location, $ r_1 - r_2 $}
	\end{subfigure}
	\hfill
	\begin{subfigure}[b]{0.325\textwidth}
		\centering 
		\includegraphics[scale=0.45]{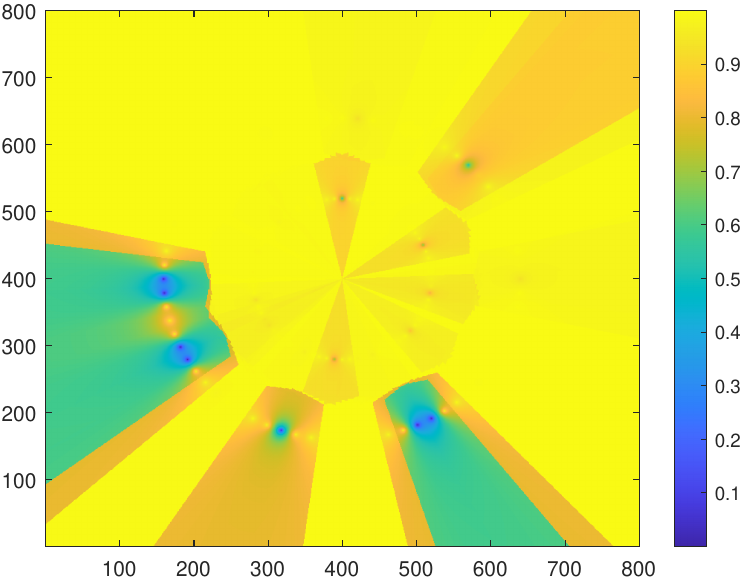}
		\caption{The baseline LSM by $ K $-nearest neighbours, $ r_1 - r_2 $.}
		\label{F:KNNbaseline LSM, $ r_1 - r_2 $}
	\end{subfigure}
	\hfill
	\begin{subfigure}[b]{0.325\textwidth}
		\centering 
		\includegraphics[scale=0.45]{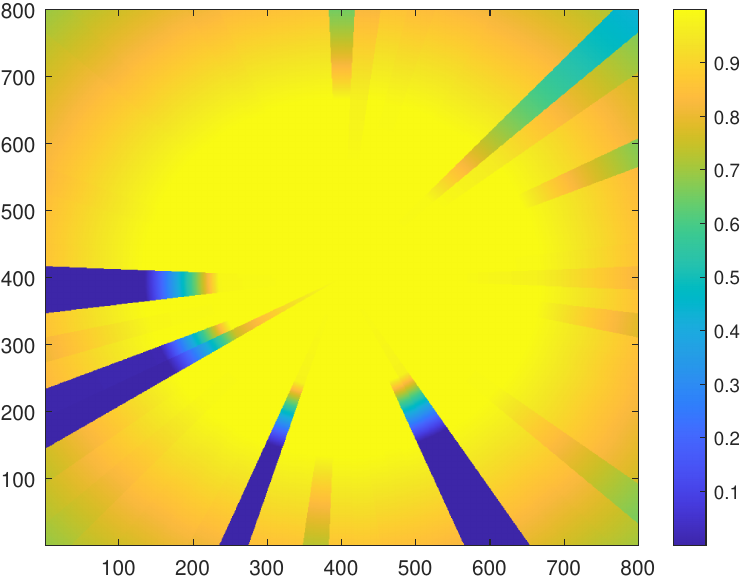}
		\caption{The Proposed LSM by all spatial correlation, $ r_1 - r_2 $.}
		\label{F:The Proposed LSM by the Spatial Correlation interpolation, $ r_1 - r_2 $}
	\end{subfigure}
	\hfill
	\begin{subfigure}[b]{0.325\textwidth}
		\centering 
		\includegraphics[scale=0.45]{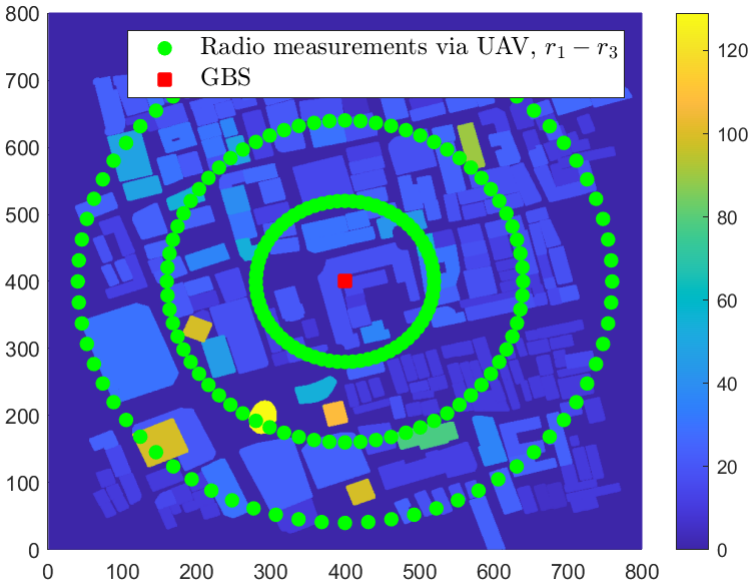}
		\caption{The sampled location by UAV from top view, $ r_1 - r_3 $.}
		\label{F:sampled location, $ r_1 - r_3 $}
	\end{subfigure}
	\hfill
	\begin{subfigure}[b]{0.325\textwidth}
		\centering 
		\includegraphics[scale=0.45]{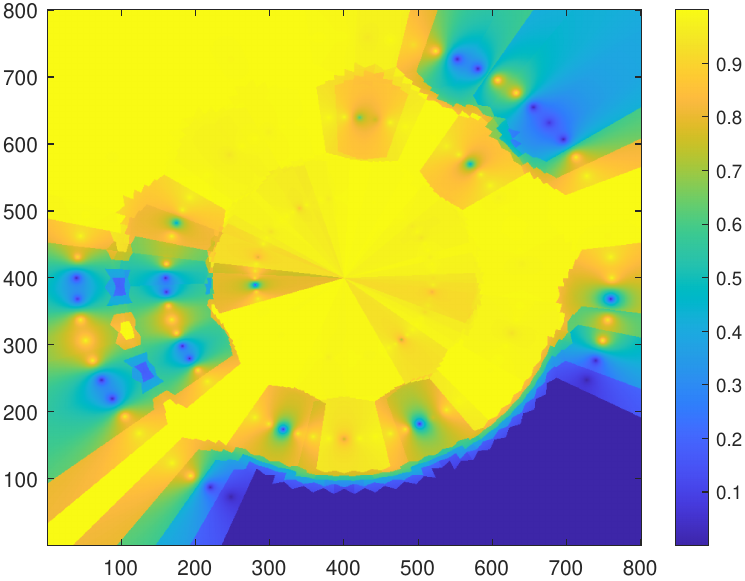}
		\caption{The baseline LSM by $ K $-nearest neighbours, $ r_1 - r_3 $.}
		\label{F:KNNbaseline LSM, $ r_1 - r_3 $}
	\end{subfigure}
	\hfill
	\begin{subfigure}[b]{0.325\textwidth}
		\centering 
		\includegraphics[scale=0.45]{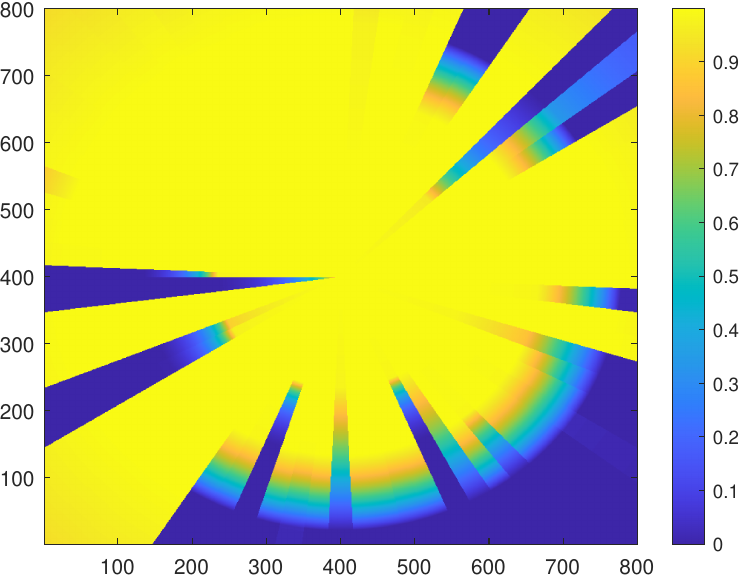}
		\caption{The Proposed LSM by all spatial correlation, $ r_1 - r_3 $.}
		\label{F:The Proposed LSM by the Spatial Correlation interpolation, $ r_1 - r_3 $}
	\end{subfigure}
	\hfill

	\caption{LSM and Probabilistic LSM constructed by our algorithm, $ \delta_D = 120$m, $ \delta_{\phi} = \frac{\pi}{36} $, $ \delta_d = 1 $ m, $ \phi_{th} = \frac{\pi}{9} $.  } 
	\label{F:LSM and EstLSM}
\end{figure*}

As shown in Fig.~\ref{F:LSM and EstLSM},  we visualize the reconstructed LSM by the proposed binary Bayesian filter algorithm exploiting the spatial correlation when we use different numbers of measurements.
We consider the simple circle trajectory to sample the radio measurement by UAV from GBS. Assume that the channel parameters within the area of interest are constant, so that all the channel gain measurements within the region can be used for LSM construction.
As the UAV flies multiple circular trajectories to perform channel gain measurements, we continuously update the LSM. We sample channel measurements at angular intervals of $ \delta_{\phi} = \frac{\pi}{36} $ at a time along the trajectory of a circle.
The radius distance $ \delta_{D} $ between two adjacent circular trajectory flights is equal, $ \delta_{D} $ is set to $ 120 $m. 
Other parameters settings are as follows: $ \phi_{th} = \frac{\pi}{9} $, $ \beta = 1 $, $ \delta_d = 1 $m.
Fig.~\ref{F:LSM and EstLSM} shows the  LSM constructed by two benchmark schemes and the proposed method under the different measurement location settings.
The simulation results are consistent with our expectation, as shown in Fig.~\ref{F:LSM and EstLSM}.  It can be seen that the proposed algorithm has best performance compared with the baseline. It is because we exploit the spatial correlation in both direction and distance to enhance the LSM construction. Besides, when only exploiting the spatial correlation in distance,  the LSM on the direction with the channel measurement is updated, which improves the performance of KNN interpolation because it extends the size of the KNN interpolation training set. 





Then, to further evaluate the performance of the proposed LSM construction
algorithm, the performance comparisons about the sampled interval of distance and angle are discussed in this part.
We evaluate the system performance measured by the MAE defined as \eqref{eq:MAE},
where the summation there is evaluated over all locations $ \mathbf{x} $ within an $ 800 $m $\times$ $ 800 $m
uniform rectangular area.
We examine the map construction accuracy of the proposed algorithm for different urban areas. 
We run Monte Carlo experiments to illustrate the effectiveness of our proposed algorithm. The number of Monte Carlo simulations is $ M_c = 5 $ in each building map, and the number of building maps is $ N_h = 5 $. We sample channel measurements of multiple circular flight trajectories around the projection of GBS in $ \mathcal{X}_h $.
The radius of the circle with the smallest flight radius among them is denoted as $ \delta_{D} $, and their radius increases with equal distance values $ \delta_{D} $ until the UAV trajectory reaches the edge of $ \mathcal{X}_h $ and the UAV obtains the channel measurements with the same angular interval $ \delta_{\phi} $ on each circular flight trajectory.

\begin{figure}[htbp]
	\centering
		\includegraphics[scale=0.6]{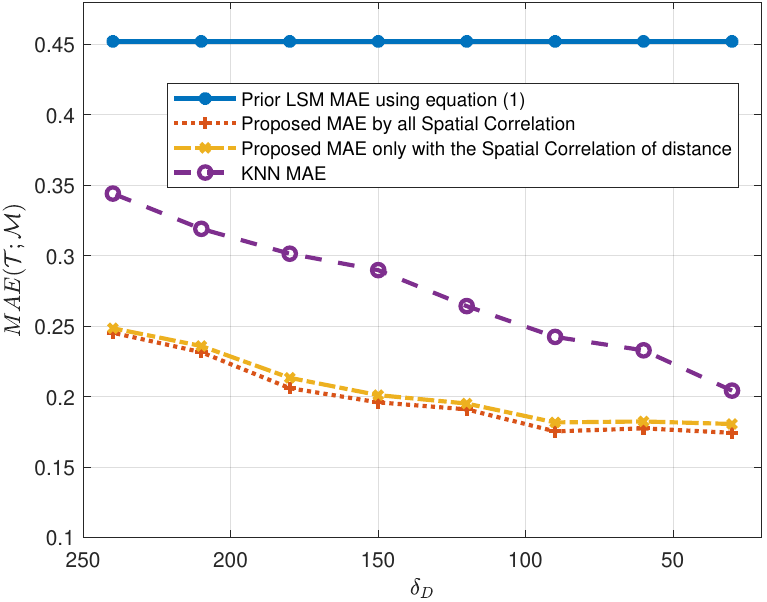}
	\caption{The performance of LSM construction algorithm $ MAE(\mathcal{T};\mathcal{M}) $ - The interval of  sample distance $ \delta_D $, \\
		parameter setting: $ \beta = 0.5, \phi_{th} = \frac{\pi}{9}, \delta_{\phi} = \frac{\pi}{9},\delta_{d} = 1 $ m.  }
	\label{F:MAE-D}
\end{figure}

Fig.~\ref{F:MAE-D} shows the relationship between the MAE of LSM constructed by the channel measurement sampled from multiple circular flight trajectories around the projection of GBS in $\mathcal{X}_h $ and the radius difference of the neighbouring circular trajectories $ \delta_{D} $. And their interval of sample angles $ \delta_{\phi} $ remains constant.  It can be observed that the MAE performance of our proposed algorithm continues to improve as we decrease the sample interval $ \delta_{D} $ over the distance.

\begin{figure}[htbp]
	\centering
		\includegraphics[scale=0.60]{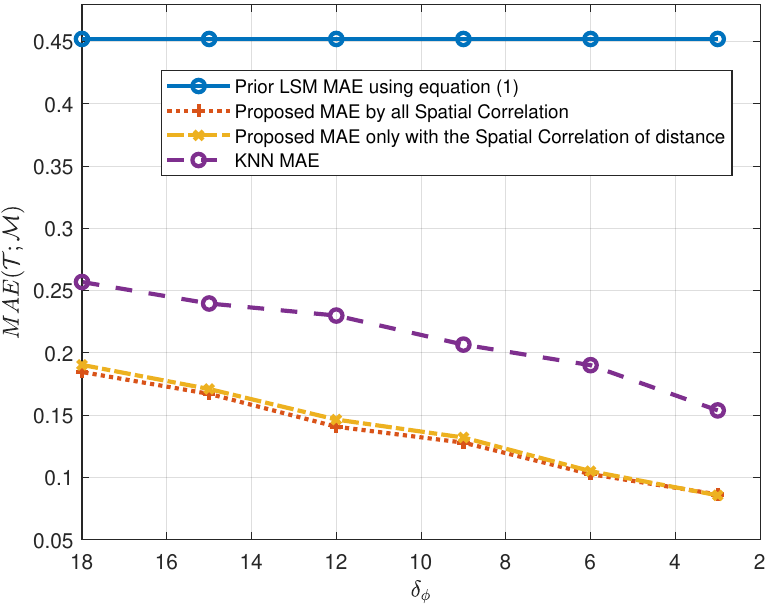}
	\caption{The performance of LSM construction algorithm $ MAE(\mathcal{T};\mathcal{M}) $ - The interval of sample angle $ \delta_{\phi} $,\\
		parameter setting: $ \beta = 0.5, \phi_{th} = \frac{\pi}{9}, \delta_{D} = 100$m, $ \delta_{d} = 1 $m. }
	\label{F:MAE-T}
\end{figure}

Further consider the MAE of LSM construction for the case where the radius of multiple circular flight trajectories around the projection of GBS is determined, while the angle interval $ \delta_{\phi} $ of sample is uncertain.
Fig.~\ref{F:MAE-T} shows the MAE performance of LSM construction algorithm versus the angle interval  $ \delta_{\phi} $ of sample. From Fig.~\ref{F:MAE-T}, we observe that LSM can be accurately reconstructed when using our proposed algorithm with the measurement settings meeting $ \delta_{\phi} = \frac{\pi}{60} $, $ \delta_{D} = 100 $m. After using the proposed algorithm, the MAE decreases from $ 0.15 $ under the KNN baseline to about $ 0.08 $, and the constructed LSM is very close to the ground-truth LSM.
And the gap between the MAE under exploiting the spatial correlation of distance and angle and that under only exploiting the spatial correlation of distance would gradually vanish as the sample interval $ \delta_{\phi} $ over the angle diminishes. 

\section{Conclusion}
This paper proposed a radio propagation spatial correlation-enhanced probabilistic LSM construction approach towards environment-aware communications and sensing.
The proposed approach has established an efficient binary Bayesian filter algorithm through environment spatial correlation to construct the probabilistic LSM. 
We further derived the spatial correlation model of the LoS probability for the location pairs on the distance and angular, respectively.
The numerical results demonstrate that our construction algorithm with environmental semantics is highly efficient for data reduction. 
 Some important lessons learnt
from this study are summarized below:
\begin{itemize}
	\item When the spatial correlation model gets accurate,
	it can improve the performance of LSM construction under a limited number of measurements. 
	\item The performance gain of the proposed binary Bayesian filter algorithm by exploiting the spatial correlation is larger when the number of directions with measurements is sparse.
	
\end{itemize}

Through this study, we have noticed that the accurate spatial correlation is essential for model-based CKM construction and the measurement location selection strategy has an impressive impact on the map construction quality. In the future,  more accurate spatial correlation model or environmental semantics will be considered to  construct other CKMs, such as channel gain map, beam index map, etc.   
The mutual information-based location selection will also be considered. Besides, fusing physical sensing measurement and channel measurements to construct CKM is also promising in the cellular-connected UAV scenarios, since it exploits the information from the multi-sources data. 


\begin{appendices}       
	\section{	Binary Bayesian filter for LSM mapping} \label{Appendix A}     
	  	Based on the Bayesian principle, we can write the posterior probability at the $ n $-th as
	\begin{equation} \label{append eq1}
		\begin{aligned}\Pr\left( l(\mathbf{x})\mid \textbf{Z}_{n} \right)=&\frac{\Pr(z_n\mid l(\mathbf{x}),\textbf{Z}_{n-1})\Pr\left(\left.l(\mathbf{x})\mid \textbf{Z}_{n-1}\right.\right)}{\Pr(z_n\mid \textbf{Z}_{n-1})} \\
			&=\frac{\Pr(z_n\mid l(\mathbf{x}))\Pr(l(\mathbf{x}) \mid \textbf{Z}_{n-1})}{\Pr(z_n\mid \textbf{Z}_{n-1})}.
		\end{aligned}
	\end{equation}
	The second "=" is due to the fact that we assume the measurement $ z_n $ is independent of the past measurements $ \mathbf{Z}_{n-1} $. i.e., conditional independence of the measurement, $ \Pr(z_n |l(\mathbf{x}),\mathbf{Z}_{n-1}) $ = $ \Pr(z_n | l(\mathbf{x})) $.
	 
	Apply Bayes' therom to $ \Pr(z_n\mid l(\mathbf{x})) $: 
	\begin{equation}\label{append eq2}
		\begin{aligned}
			\Pr(z_n\mid l(\mathbf{x}))=\frac{\Pr\left(l(\mathbf{x})\mid z_n\right)\Pr\left(z_n\right)}{\Pr(l(\mathbf{x}))}.
		\end{aligned}
	\end{equation}
	Subsituting to eq.\eqref{append eq1},  we have 
	\begin{equation} \label{append eq3}
		\begin{aligned}& 
			\Pr\left( l(\mathbf{x})\mid \textbf{Z}_{n} \right) \\
			&=\frac{\Pr\left(l(\mathbf{x})\mid z_n)\right)\Pr(z_n)\Pr\left( l(\mathbf{x})\mid \textbf{Z}_{n-1} \right)}{\Pr(l(\mathbf{x}))\Pr(z_n\mid \textbf{Z}_{n-1})}.
		\end{aligned}
	\end{equation}
	Similarly, for antagonistic events  $ \bar{l(\mathbf{x})}$ , we have 
		\begin{equation} \label{append eq4}
		\begin{aligned}
		&\Pr\left(\bar{l(\mathbf{x})}\mid \textbf{Z}_{n} \right) \\
		&=\frac{\Pr\left(\bar{l(\mathbf{x})} \mid z_n\right)\Pr(z_n)\Pr\left( \bar{l(\mathbf{x})}\mid \textbf{Z}_{n-1} \right)}{\Pr( \bar{l(\mathbf{x})})\Pr(z_n\mid \textbf{Z}_{n-1})}.
		\end{aligned}
	\end{equation}
	Eq. \eqref{append eq3} divided by eq. \eqref{append eq4} has
		\begin{align}
			&\frac{\Pr\left( l(\mathbf{x})\mid \textbf{Z}_{n} \right)}{ \Pr\left(\bar{l(\mathbf{x})}\mid \textbf{Z}_{n} \right)}\\ \nonumber
      & =\frac{\Pr\left(\left.l(\mathbf{x}) \right|z_n\right)}{\Pr\left(\left.\bar{l(\mathbf{x})}\mid z_n\right)\right.}\frac{\Pr\left(\left.l(\mathbf{x}) \right|\textbf{Z}_{n-1}\right)}{\Pr\left(\left.\bar{l(\mathbf{x})}\right|\textbf{Z}_{n-1}\right)}\frac{\Pr\left(\bar{l(\mathbf{x})}\right)}{\Pr\left(l(\mathbf{x})\right)}  \label{eq:binaryfilter} \\ 
			&=\frac{\Pr\left(l(\mathbf{x}) \mid z_n\right)}{1-\Pr\left(l(\mathbf{x}) \mid z_n\right)}\frac{\Pr\left(l(\mathbf{x}) \mid \textbf{Z}_{n-1}\right)}{1-\Pr\left(l(\mathbf{x}) \mid \textbf{Z}_{n-1}\right)}\frac{1-\Pr(l(\mathbf{x}))}{\Pr\left(l(\mathbf{x})\right)}. \nonumber
		\end{align}
	We then take the algorithm to get eq.\eqref{eq:recursiveLn}.
	 \section{The derivation of inverse measurement model $ \Pr(l(\mathbf{x})|z_n) $} \label{Appendix B}
	 The equation \eqref{eq:updateLocNM} can be computed as follows: 
	
	 	\begin{equation}\label{eq: not UAV position PDF}
	 		\begin{aligned}
	 			\Pr(l(\mathbf{x}) |z_n)
                =& \frac{\Pr(z_n,l(\mathbf{x}))}{\Pr(z_n)} \\
	 			=& \frac{ \sum_{l(\mathbf{x}_n)} \Pr(z_n, l(\mathbf{x}), l(\mathbf{x}_n))}{\Pr(z_n)} \\
	 			=& \frac{ \sum_{l(\mathbf{x}_n)} \Pr(z_n| l(\mathbf{x}), l(\mathbf{x}_n))\Pr(l(\mathbf{x}), l(\mathbf{x}_i))}{\Pr(z_n)} \\
	 			=& \frac{\Pr(z_n| l(\mathbf{x}), l(\mathbf{x}_n) = 1)p( l(\mathbf{x}), l(\mathbf{x}_n) =1) }{\Pr(z_n)} \\
    &+ \frac{\Pr(z_n| l(\mathbf{x}), l(\mathbf{x}_n) = 0)\Pr( l(\mathbf{x}), l(\mathbf{x}_n) =0)}{\Pr(z_n)}
    \\
	 			\overset{\text{ex}}{=}& \frac{\Pr(z_n| l(\mathbf{x}_n) = 1)\Pr( l(\mathbf{x}), l(\mathbf{x}_n) =1)}{\Pr(z_n)} \\
    &+ \frac{ \Pr(z_n| l(\mathbf{x}_n) = 0)\Pr( l(\mathbf{x}), l(\mathbf{x}_n) =0)}{\Pr(z_n)}
    \\
	 			=& \frac{\Pr( l(\mathbf{x}_n) = 1| z_n  )\Pr(l(\mathbf{x}), l(\mathbf{x}_n) = 1)}{\Pr(l(\mathbf{x}_n) = 1)}
	 			\\
                &+ \frac{\Pr( l(\mathbf{x}_n) = 0| z_n  )\Pr(l(\mathbf{x}), l(\mathbf{x}_n) = 0)}{\Pr(l(\mathbf{x}_n) = 0)} \\
	 			=& \Pr( l(\mathbf{x}_n) = 1| z_n  )\Pr(l(\mathbf{x})| l(\mathbf{x}_n) =1) \\
                &+ \Pr( l(\mathbf{x}_n) = 0|z_n) \Pr(l(\mathbf{x})| l(\mathbf{x}_n) = 0)).
	 		\end{aligned}
	 	\end{equation}
	 The fifth "$ \overset{\text{ex}}{=} $" of the above equation is due to the fact that the measurement $ z_n $ is only relevant to the link state $ l(\mathbf{x}_n) $ at its location, and is independent of other locations, i.e $ \Pr(z_n \mid l(\mathbf{x}),l(\mathbf{x}_n) = j) = \Pr(z_n \mid l(\mathbf{x}_n) = j), j = 0 $ or $ 1 $.
\section{The derivation of spatial correlation $ r_{ij}(\mathbf{x},\mathbf{x}_n) $} \label{Appendix C}
The spatial correlation of LoS probability $ r_{ij}(\mathbf{x},\mathbf{x}_n) $ can be divided into two cases:
\begin{itemize}
	\item The spatial correlation of location pairs with the same distance but different azimuth directions.
	\item  The spatial correlation of location pairs with the same azimuth directions but different distance.
\end{itemize}

For the first case, we need to derive the eq. \eqref{eq:CorrelationOnAngle}.
According to the radio propagation rule, we can easily obtain
\begin{equation}
	\begin{aligned}
		r_{11}(\mathbf{x},\mathbf{x}_n) = 1, r < r_n.  \\
		r_{00}(\mathbf{x},\mathbf{x}_n) = 1, r > r_n.
	\end{aligned}
\end{equation}
To calculate the other two formulas in \eqref{eq:CorrelationOnAngle}, we exploit the Bayesian principle to solve them. When $ r < r_n $, $ r_{00}(\mathbf{x}_n,\mathbf{x}) = 1 $
\begin{equation}
	\begin{aligned}
		r_{00}(\mathbf{x}_n,\mathbf{x}) &= \frac{\Pr(l(\mathbf{x})=0,l(\mathbf{x}_n)=0)}{\Pr(l(\mathbf{x}_n)=0)} \\
		&=\frac{r_{00}(\mathbf{x}_n,\mathbf{x})\Pr (l(\mathbf{x})=0)}{\Pr(l(\mathbf{x}_n)=0)} \\
		&=\frac{\Pr (l(\mathbf{x})=0)}{\Pr(l(\mathbf{x}_n)=0)}.
	\end{aligned}
\end{equation} 
When $ r > r_n $, $ r_{11}(\mathbf{x}_n,\mathbf{x})= 1 $,
\begin{equation}
	\begin{aligned}
				r_{11}(\mathbf{x}_n,\mathbf{x}) &= \frac{\Pr(l(\mathbf{x})=1,l(\mathbf{x}_n)=1)}{\Pr(l(\mathbf{x}_n)=1)} \\
		&=\frac{r_{11}(\mathbf{x}_n,\mathbf{x})\Pr (l(\mathbf{x})=1)}{\Pr(l(\mathbf{x}_n)=1)} \\
		&=\frac{\Pr (l(\mathbf{x})=1)}{\Pr(l(\mathbf{x}_n)=1)}.
	\end{aligned}
\end{equation}

For the second case, we need to derive the eq. \eqref{eq:Correlationondistance}.
The joint probability of $ l(\mathbf{x}) $ and $ l(\mathbf{x}_n) $ can be decomposed as 
\begin{align}\label{eq:jointtocondi}
	p_{ij}  = r_{ij}(\mathbf{x},\mathbf{x}_n) \Pr(l(\mathbf{x}_n)=j).
\end{align}

According to the definition \eqref{eq:phi} of phi coefficient and \eqref{eq:jointtocondi}, we have 
\begin{equation}\label{eq:phia}
	\rho = \frac{\Pr(l(\mathbf{x_n})=1)\Pr(l(\mathbf{x}_n)=0)(r_{11}r_{00} - r_{10}r_{01})}{\sqrt{\Pr(l(\mathbf{x})=0)\Pr(l(\mathbf{x}_n)=0)\Pr(l(\mathbf{x}_n)=1)\Pr(l(\mathbf{x})=1)} }.
\end{equation}
In addition, $ r_{ij} $ has the following equation,
\begin{equation}
	\begin{aligned}\label{eq:relation1}
		&r_{10} + r_{00} = 1; \\
		&r_{11} + r_{01} = 1; 
	\end{aligned}
\end{equation}
\begin{equation}
	\begin{aligned}\label{eq:relation2}
			r_{11} &= \frac{\Pr(l(\mathbf{x})=1,l(\mathbf{x}_n)=1)}{\Pr(l(\mathbf{x}_n)=1) } \\
		&= \frac{\Pr(l(\mathbf{x}) = 1)) - r_{10} \Pr(l(\mathbf{x}_n)=0)}{\Pr(l(\mathbf{x}_n)=1)}.
	\end{aligned}
\end{equation}

Finally, we take eq. \eqref{eq:relation1} and eq. \eqref{eq:relation2} into eq. \eqref{eq:phia}, and after a simple calculation we have eq. \eqref{eq:Correlationondistance}.

\end{appendices}

\bibliographystyle{IEEEtran}
\bibliography{CKM}

\end{document}